%% file: main.tex
\documentclass[10pt, journal]{IEEEtran}
\pdfoutput=1
\usepackage{graphicx}
\usepackage{amsmath}
\usepackage{amssymb}
\usepackage{amsthm}
\usepackage{algpseudocode}
\usepackage{algorithm}
\usepackage{setspace}
\usepackage{amsfonts}
\usepackage{subfigure}
\usepackage{setspace}
\usepackage{color}
\newcommand{\argmax}{\operatornamewithlimits{argmax}}

\newcommand{\blue}[1]{{\color{black}{#1}}}
\algnewcommand\algorithmicinput{\textbf{Input:}}
\algnewcommand\Input{\item[\algorithmicinput]}

\newtheorem{theorem}{Theorem}[section]
\newtheorem{lemma}[theorem]{Lemma}

\newtheorem{corollary}[theorem]{Corollary}
\newtheorem{fact}[theorem]{Fact}
\newtheorem{definition}[theorem]{Definition}

\begin{document}

\author{Pan~Lai,~
        Rui~Fan,~
        Xiao~Zhang,~%\IEEEmembership{Member,~IEEE,}
				Wei~Zhang,~%\IEEEmembership{Member,~IEEE,}
				Fang~Liu,~%\IEEEmembership{Member,~IEEE}
				Joey~Tianyi~Zhou~%\IEEEmembership{Member,~IEEE}
\thanks{An earlier version of this work has been presented at IEEE International Parallel and Distributed Processing Symposium (IPDPS), 2016 \cite{PanAssign}.}
\thanks{
Pan Lai is with College of Computer Science, South-Central University for Nationalities, and also with the Engineering Systems and Design Pillar, Singapore University of Technology and Design, Singapore. (e-mail: plai1@ntu.edu.sg). Rui Fan is with School of Information Science and Technology, ShanghaiTech University, Shanghai, China.
(e-mail: fanrui@shanghaitech.edu.cn). Xiao Zhang is with College of Computer Science, South-Central University for Nationalities, Wuhan, China.
(e-mail: xiao.zhang@my.cityu.edu.hk). Wei Zhang is with the Information and Communications Technology Cluster, Singapore Institute of Technology, Singapore. 
(e-mail: wei.zhang@ieee.org). Fang Liu is with School of Science and Technology, Singapore University of Social Sciences, Singapore. (e-mail: fliu@ieee.org). Joey Tianyi Zhou is with Institute of High Performance Computing, A*STAR, Singapore. (e-mail: joey.tianyi.zhou@gmail.com).}% <-this % stops a space
%\thanks{Rui Fan is with School of Information Science and Technology, ShanghaiTech University, Shanghai, China.
%(e-mail: fanrui@shanghaitech.edu.cn).}
%\thanks{Xiao Zhang is with College of Computer Science, South-Central University for Nationalities, Wuhan, China.
%(e-mail: xiao.zhang@my.cityu.edu.hk).}
%\thanks{Wei Zhang is with the Information and Communications Technology Cluster, Singapore Institute of Technology, Singapore. 
%(e-mail: wei.zhang@ieee.org).}
%\thanks{Fang Liu is with School of Science and Technology, Singapore University of Social Sciences, Singapore. (e-mail: fliu@ieee.org).}
%\thanks{Joey Tianyi Zhou is with Institute of High Performance Computing, A*STAR, Singapore. (e-mail: joey.tianyi.zhou@gmail.com).}
}

\title{Utility Optimal Thread Assignment and Resource Allocation in \blue{Multi-Server Systems}}
\maketitle

\input{abstract}
\input{introduction}
\input{related-works}
\input{model}
\input{NP-hard}
\input{algorithm}

\input{fast-algorithm}
\input{nonconcave}
\input{experiment}
\input{multi-types-dp}
\input{conclusion}
\input{bib}
\input{author}

\end{document}

%% file: abstract.tex
\begin{abstract}
Achieving high performance in many \blue{multi-server systems} requires finding a good assignment of worker threads to servers and also effectively allocating each server's resources to its assigned threads. The assignment and allocation components of this problem have been studied extensively but largely separately in the literature. In this paper, we introduce the \emph{assign and allocate (AA)} problem, which seeks to simultaneously find an assignment and allocation that maximizes the total utility of the threads. Assigning and allocating the threads together can result in substantially better overall utility than performing the steps separately, as is traditionally done. We model each thread by a utility function giving its performance as a function of its assigned resources.  We first prove that the AA problem is NP-hard.  We then present a $2 (\sqrt{2}-1) > 0.828$ factor approximation algorithm for concave utility functions, which runs in $O(mn^2 + n (\log mC)^2)$ time for $n$ threads and $m$ servers with $C$ amount of resources each. We also give a faster algorithm with the same approximation ratio and $O(n (\log mC)^2)$ time complexity.  We then extend the problem to two more general settings.  First, we consider threads with nonconcave utility functions, and give a 1/2 factor approximation algorithm.  Next, we give an algorithm for threads using multiple types of resources, and show the algorithm achieves good empirical performance.  We conduct extensive experiments to test the performance of our algorithms on threads with both synthetic and realistic utility functions, and find that \blue{they achieve} over 92\% of the optimal utility on average. We also compare our \blue{algorithms} with a number of practical heuristics, and find that \blue{our algorithms} achieve up to 9 times higher total utility. 

\end{abstract}

%% file: introduction.tex
\section{Introduction}
In this paper, we study efficient ways to execute a set of resource constrained worker threads on multiple servers.  Our problem consists of two steps.  First, each thread is \emph{assigned} to a server. Subsequently, the resources at each server are \emph{allocated} to the threads assigned to it. Each thread obtains a certain utility based on the resources it is allocated, which is captured in the form of a utility function. The goal is to maximize, over all possible assignments and allocations, the total utility of all the threads. We call this problem \emph{AA}, for \emph{assign and allocate}.  

%$\hat{c}_1$, $\hat{c}_2$
The AA problem can model a range of \blue{system settings}. \blue{One example is a web hosting center, where a service provider operates websites on behalf of customers.  Incoming web requests are serviced by threads which are run on center's servers.    As requests have different characteristics, the performance of the thread for each request depends on the amount of resources, such as processing or memory, which the thread is allocated.  The goal of the service provider is to process the largest number of requests, and it does this by controlling both the servers which threads are assigned to and the resources allocated to the threads on each server, in a way which maximizes the overall system performance.  Another application of the AA problem is in cloud computing.  Clouds are a promising paradigm for providing configurable computing resources to users. A cloud provider sells virtual machine instances (corresponding to threads in AA) running on physical machines (corresponding to servers). Customers use utility functions to express their willingness to pay for instances consuming different amounts of resources, and the provider's task is to assign and size the virtual machines to maximize her profit (utility).  As a final application, consider a multicore processor, where each core corresponds to a server offering its shared cache as a resource to concurrently executing threads. Each thread is first bound to a core, after which cache partitioning \cite{Suh, Qureshi} can enforce an allocation of the core's cache among the assigned threads. A thread's performance is often strongly dependent on its cache allocation \cite{Qureshi, Lai, PanMakespan}. A scheduler tries to maximize overall system performance through an efficient mapping of threads to the cores and effectively partitioning of each core's cache.}

\blue{The two steps in AA correspond to the \emph{thread assignment} and \emph{resource allocation} problems, both of which have been studied extensively in the literature. However, to the best of our knowledge, these problems have not been studied together in the unified context considered in our work. Existing works on resource allocation \cite{Suh, Qureshi, Lai, Fox, Galil} largely deal with dividing the resources on a single server among a given set of threads. It is not clear how to apply these algorithms when there are multiple servers, since there are many possible ways to initially assign the threads to servers, and certain assignments result in low overall performance regardless of how resources are subsequently allocated. For example, if there are two types of threads, one with high maximum utility and one with low utility, then assigning all the high utility threads to the same server will result in competition between them and depressed overall utility no matter how resources are allocated.  Likewise, existing works on thread assignment \cite{Urgaokar2, Karve, Becchi, Radojkovic} often overlook the resource allocation aspect. Typically in these works each thread requests a fixed amount of resource. Once assigned to a server, a thread is allocated precisely the resources it requested, without any adjustments made based on the requests of other threads assigned to the same server. This can also lead to suboptimal performance. For example, consider a thread which obtains $x^\beta$ utility when allocated $x$ amount of resource, for some $\beta \in (0,1)$\footnote{For $\beta = \frac{1}{2}$, this is known as the ``square root rule'' \cite{Srinivasan, Thiebaut}.}, and suppose the thread requests $z > 0$ resources. Then when there are $n$ threads and one server with $C$ resources, typical thread assignment algorithms would give $\frac{C}{z}$ threads $z$ resources each while the rest receive 0, leading to a total utility of $Cz^{\beta-1}$; note that this quantity is constant in $n$. However, the optimal allocation gives $\frac{C}{n}$ resources to each thread and has total utility $C^\beta n^{1-\beta}$, which is arbitrarily better than the first allocation for large $n$. To see this, consider an example of $C=100, z=10, \beta=0.5$, then the ratio between $C^\beta n^{1-\beta}$ and $Cz^{\beta-1}$ is $\sqrt{\frac{n}{10}}$, which is arbitrarily large for large $n$.}

\blue{The AA problem models each thread using a nondecreasing \emph{utility function} giving its performance as a function of the resources it receives. In practice, many utility functions are concave, capturing a frequently observed diminishing returns property \cite{Qureshi}. However, in certain settings, such as cache hit rates for different amounts of cache allocation, the utility function may also be nonconcave \cite{Lai}.  Furthermore, a thread's utility may sometimes depend on multiple types of resources.  The goal in AA is to simultaneously find assignments and allocations for all the threads which maximizes their total utility.  To the best of our knowledge, we are the first to study the two problems in a unified context. Our paper makes the following contributions.}

\begin{enumerate}
	\item We show that the AA problem is NP-hard, even when there are only two servers. In contrast, the problem is efficiently solvable when there is a single server \cite{Galil}.
	\item We present an approximation algorithm which achieves at least $\alpha = 2(\sqrt{2} - 1) > 0.828$ times the optimal utility for concave utility functions.  The algorithm relates the optimal solution of a single server problem to an approximately optimal solution of the multiple server problem.  It runs in  $O(mn^2 + n (\log mC)^2)$ time, where $n$ and $m$ are the number of threads and servers, respectively, and $C$ is the amount of resource on each server.  We also present a faster algorithm with $O(n (\log mC)^2)$ running time and the same approximation ratio.
	\item We also consider threads with nonconcave utility functions, and present an algorithm with approximation ratio $\frac{1}{2}$ and running time $O(s nmC \alpha(mC) (\log mC)^2)$, where $s$ is the maximum number of concave or convex segments in each utility function, and $\alpha$ is the inverse Ackermann function.
	\item \blue{While the previous three algorithms consider utility functions based on a single resource type, we also present an algorithm for utility functions based on multiple types of resources.  We show that this algorithm achieves good empirical performance.}  
	\item We conduct extensive experiments to test the performance of our algorithms.  We use several types of synthetic and realistic utility functions, and show that our algorithms obtains over 92\% of the maximum utility on average.  We also compare our algorithms with several simple but practical heuristics, and show that they achieve up to 9 times higher utility for very heterogeneous threads.   
\end{enumerate}     

The rest of paper is organized as follows. In Section II, we describe related works on thread assignment and resource allocation.  Sections III formally defines our model and the AA problem. Section IV proves AA is NP-hard. Section V presents an approximation algorithm for concave utility functions and its analysis, and Section VI proposes a faster algorithm. Section VII presents and analyzes an algorithm for nonconcave utility functions. In Section VIII, we describe our experimental results. Section IX extends our proposed algorithms to multiple resource types. Finally, Section X concludes and discusses some future problems.

%% file: related-works.tex
\section{Related Works}
There is a large body of work on resource allocation for a single server.  Fox \emph{et al.} \cite{Fox} considered concave utility functions and proposed a greedy algorithm to find an optimal allocation in $O(nC)$ time, where $n$ is the number of threads and $C$ is the amount of resource on the server. Galil \cite{Galil} proposed an improved algorithm with $O(n (\log C)^2)$ running time, by doing a binary search to find an allocation in which the derivatives of all the threads' utility functions are equal, and the total resources used by the allocation is $C$. Resource allocation for nonconcave utility functions is weakly NP-complete. However, Lai and Fan \cite{Lai} identified a structural property of real-world utility functions which leads to fast parameterized algorithms.  %In a multiprocessor setting, cache hit rates are often modeled by utility functions.  Hit rate curves can be determined by running threads multiple times using different cache allocations.  Qureshi \emph{et al.} \cite{Qureshi} proposed efficient hardware based methods to minimize the overhead of this process, and also designed algorithms to partition a shared cache between multiple threads using methods based on \cite{Fox}.  
%Zhao \emph{et al.} \cite{Zhao} proposed software based methods to measure page misses and allocate memory to virtual machines in a multicore system.  Chase \emph{et al.} \cite{Chase} used utility functions to dynamically determine resource allocations among users in hosting centers and maximize total profit.

%There has also been extensive work on assigning threads to cores in multicore architectures.  Becchi \emph{et al.} \cite{Becchi} proposed a scheme to determine an optimized assignment of threads to heterogeneous cores. They characterized the behavior of each thread on a core by a single value, its IPC (instructions per clock), independent of the amount of resource the thread uses.  In contrast, we model a thread using a utility function giving its performance for different resource allocations.  Radojkovi\'{c} \emph{et al.} \cite{Radojkovic} proposed a statistical approach to assign threads on massively multithreaded processors, choosing the best assignment out of a large random sample. The results in \cite{Becchi} and \cite{Radojkovic} are based on simulations and provide no theoretical guarantees.  

\blue{Our work is related to the \emph{application placement} problem, in which applications with different resource requirements need to be mapped to servers while fulfilling certain quality of service guarantees. Urgaokar \emph{et al.} \cite{Urgaokar2} proposed offline and online approximation algorithms for application placement, and the offline algorithm achieves a $\frac{1}{2}$ approximation ratio. \cite{Karve} proposed algorithms to place web applications on servers with the goal of maximizing the amount of demand which can be satisfied. They model each thread using a single value corresponding to a fixed amount of allocated resource, instead of a utility function allowing a range of resource allocations as our paper does.}

%\emph{Co-scheduling} is a technique which divides a set of threads into subsets and executes each subset together on one chip of a multicore processor.  It can be used to minimize cache interference between threads. Jiang et al. \cite{Jiang} proposed algorithms to find optimal co-schedules for pairs of threads, and also approximation algorithms for co-scheduling larger groups of threads. Tian \emph{et al.} \cite{Tian1} also proposed exact and approximation algorithms for co-scheduling on chip multiprocessors with a shared last level cache. Zhuralev \emph{et al.} \cite{Zhura} surveyed scheduling techniques for sharing resources in multicores. The works on co-scheduling require measuring the performance from running different groups of threads together.  When co-scheduling large groups, the number of measurements required becomes prohibitive. In contrast, we model threads by utility functions, which can be efficiently determined by measuring the performance of individual threads instead of groups.   

The AA problem is also related to the multiple knapsack and multiple-choice knapsack (MCKP) problems, for which there has been a number of studies.  For the former, Neebe \emph{et al.} \cite{Neebe} proposed a branch-and-bound algorithm, and Chekuri \emph{et al.} \cite{Chekuri} proposed a polynomial time approximation scheme.  The multiple knapsack problem differs from AA in that each item, corresponding to a thread, has a single weight and value, corresponding to a single resource allocation and associated utility.  In contrast, we use utility functions which allow threads a continuous range of allocations and utilities. The MCKP problem can model utility functions as it considers classes of items with different weights and values and chooses one item from each class; each class corresponds to a utility function.  However, MCKP only considers a single knapsack, and thus corresponds to a restricted form of AA with one server.  Kellerer \emph{et al.} \cite{Kellerer} proposed a greedy MCKP algorithm.  Lawler \cite{Lawler} proposed a $1- \epsilon$ approximate algorithm, while Gens and Levner \cite{Gens} proposed a $\frac{4}{5}$ approximate algorithm with better running time.  AA can be seen as a combined multiple-choice multiple-knapsack problem.  We are not aware of any previous work on this problem.  The model used in this paper corresponds to the case where there are items for every weight.

\blue{There has also been a large amount of work on resource provisioning for cloud computing and data centers, which are related to the multi-server setting we consider. \cite{Li, Ren} analyzed online bin packing algorithms for the problem of dispatching cloud gaming requests to servers in order to minimize total cost. Bobroff \emph{et al.} \cite{Bobroff} proposed a first-fit heuristic which dynamically places virtual machines (VMs) on physical machines to minimize the number of machines required to support a workload at a specified allowable rate of SLA violations. Jennings \emph{et al.} \cite{Jennings} surveyed many virtual machine placement schemes in clouds. These works typically follow a bin packing formulation, whereas the AA problem corresponds to a multiple-choice multiple knapsack problem. Lampe \emph{et al.} \cite{Lampe} proposed an auction scheme to allocate VMs to users to maximize a provider's profit. \cite{Shi, Mashatekhy} proposed auction schemes to allocate VMs to users to maximize social welfare. \cite{Shi} focuses on the online problem while \cite{Mashatekhy} focuses on the offline problem. In auction schemes, each user submits a bid stating their desired number of VM instances and their maximum willingness to pay, and the system provider decides whether to accept the users' bids. Each user's request for VMs is allowed to be assigned to more than one physical machine. However, in the AA problem, each thread can be assigned to only one server. Tang \emph{et al.} \cite{Tang} proposed a policy to allocate multiple resources to different users to achieve long-term fairness, but do not consider how to assign the resource requests to different servers as our paper does. Han \emph{et al.} \cite{HanCloud} proposed an algorithm to schedule (offload) jobs on mobile devices to edge servers to minimize the jobs' response times in edge-clouds. Wei \emph{et al.} \cite{Wei} proposed an online algorithm to minimize cost for data centers with multiple servers and randomly arriving service requests by determining the state of each server among three types (i.e., active, idle and setup). Han \cite{HanDataCenter} proposed an approximate dynamic VM management method to minimize power consumption in data centers.  While all these works have as a basic goal---an efficient usage of servers and resources, they usually consider these aspects separately, and also do not model thread performance using utility functions.}
%Liu \emph{et al.} \cite{Liu} proposed an approximation algorithm to provision servers to satisfy time-varying user demands to minimize the total cost of all servers in a cloud, but do not consider how to assign the resource requests to different servers. In contrast, we seek to maximize the total utility of all the threads and consider thread assignment aspect.
%Han \emph{et al.}
  
%The AA problem can be a reduction from partition problem \cite{Garey79} which will be shown later in Section IV. Many algorithms have been proposed for partition problem. However, they cannot be directly used in AA problem, since AA problem tries to maximize the total utility on all servers while the partition problem tries to balance the total utility on each server.

%% file: model.tex
\section{Model}
\label{sec-model}
In this section we formally define our model and problem.  We consider a set of $m$ homogenous servers $s_1, \ldots, s_m$. Each server has $C > 0$ amount of resources, where $C$ is a positive integer. We note that homogeneous servers, \emph{i.e.} servers with the same processing capabilities and available resources, have been widely studied in the literature \cite{Urgaokar2,Lin2011} and accurately model a number of scenarios. For example, multicore processors typically contain shared caches of the same size, and datacenters often have many identically configured servers for ease of management. We also have $n$ threads $t_1, \ldots, t_n$. The set of threads is static, to capture threads performing long-running tasks.  Let $S$ and $T$ denote the set of servers and threads, respectively.  Each thread $t_i$ is associated with a \emph{utility function} $f_i: [0,C] \rightarrow \mathbb{Z}^{\geq 0}$, giving its performance as a function of the resources it is allocated.  We assume that $f_i$ is nonnegative and nondecreasing.  We also assume that $f_i$ is either concave, or consists of a set of concave or convex segments.  In the former case, the concavity assumption models a diminishing returns property frequently observed in practice \cite{Qureshi}, and is often used to model cache and memory performance \cite{Srinivasan, Thiebaut}.  The latter case of a nonconcave function consisting of several concave or convex segments was introduced in \cite{Lai}, and can be used to model an arbitrary function $f: [0, C] \rightarrow \mathbb{Z}^{\geq 0}$.  It was observed in \cite{Lai} that most utility functions $f$, despite being possibly nonconcave, consist of a small number  $s$ of \emph{segments} $[0,b_1], (b_1, b_2], \ldots, (b_{s-1}, C]$, such that $f$ is either concave or convex in each segment $(b_i, b_{i+1}]$.  Examples of such functions include cache hit rate functions from the SPEC CPU benchmarks, such as \emph{aspi} or \emph{swim}.  

Our goal is to assign the threads to the servers in a way which respects the resource bounds and maximizes the total utility. While a solution to this problem involves both an assignment of threads and allocations of resources, for simplicity we use the term assignment to refer to both.  Thus, an \emph{assignment} is given by a vector $[(r_1, c_1), \ldots, (r_n, c_n)]$, indicating that each thread $t_i$ is allocated $c_i$ amount of resource on server $s_{r_i}$.  Let $S_j$ be the set of threads assigned to server $s_j$.  That is, $S_j = \{i \, | \, r_i = j\}$.  Then for all $1 \leq j \leq m$, we require $\sum_{i \in S_j} c_i \leq C$, so that the threads assigned to $s_j$ use at most $C$ resources.  We assume that every thread is assigned to some server, even if it receives 0 resources on the server.  The \emph{total utility} from an assignment is $\sum_{j=1}^m \sum_{i \in S_j} f_i(c_i) = \sum_{i \in T} f_i(c_i)$.  The \emph{AA (assign and allocate)} problem is to find an assignment that maximizes the total utility.

%The concavity assumption often holds in practice \cite{Qureshi, Lai} because it models a diminishing returns property in which the gains of a thread from each unit of resource it is allocated decreases with its allocation.  Also, we assume all servers are homogeneous. This also makes sense in many settings. For example, most multicore processors (e. g. 2-core Intel Core 2 Duo, 4-core AMD Opteron) have homogenous cores (servers), and homogeneous distributed systems are also widely investigated in much research work \cite{Urgaokar2, Tang}.

%% file: NP-hard.tex
\section{Hardness of the Problem}
In this section, we show that it is NP-hard to find an assignment maximizing the total utility, even when there are only two servers and the utility functions are all concave.  Thus, it is unlikely there exists an efficient optimal algorithm for the AA problem.  This motivates the approximation algorithms we present in Sections \ref{sec-alg} and \ref{sec-fast-alg}.

\begin{theorem}\label{thm-np}
	Finding an optimal AA assignment is NP-hard, even when there are only two servers and all threads have concave utility functions.  
\end{theorem}

\begin{proof}
	We give a reduction from the NP-hard \emph{partition problem} \cite{Garey79} to the concave AA problem with two servers.  In the partition problem, we are given a set of numbers $S = \{c_1, \ldots, c_n\}$ and need to determine if there exists a partition of $S$ into sets $S_1$ and $S_2$ such that $\sum_{i \in S_1} c_i = \sum_{i \in S_2} c_i$.  Given an instance of partition, we create an AA instance $A$ with two servers each with $C = \frac{1}{2} \sum_{i=1}^n c_i$ amount of resources.  There are $n$ threads $t_1, \ldots, t_n$, where the $i$'th thread has utility function $f_i$ defined by 
	\begin{equation*}
	f_i(x) = 
	\begin{cases}
	x & \text{if } x \leq c_i \\
	c_i & \text{otherwise}
	\end{cases}
	\end{equation*}
	The $f_i$ functions are nondecreasing and concave. We claim the partition instance has a solution if and only if $A$'s maximum utility is $\sum_{i=1}^n c_i$.  For the if direction, let $A^* = [(r^*_1 ,c^*_1), \ldots, (r^*_n, c^*_n)]$ denote an optimal solution for $A$, and let $S_1$ and $S_2$ be the set of threads assigned to the servers 1 and 2, respectively.  We show that $S_1, S_2$ solve the partition problem.  We first show that $c^*_i = c_i$ for all $i$.  Indeed, if $c^*_i < c_i$ for some $i$, then $f_i(c^*_i) < c_i$, while $f_j(c^*_j) \leq c_j$, for all $j \neq i$.  Thus, $\sum_{j=1}^n f_j(c^*_j) < \sum_{j=1}^n c_j$, which contradicts the assumption that $A^*$'s utility is $\sum_{j=1}^n c_j$.  Next, suppose $c^*_i > c_i$ for some $i$.  Then since $A^*$ is a valid assignment, we have $\sum_{i \in S_1} c^*_i + \sum_{i \in S_2} c^*_i \leq C + C = \sum_{i=1}^n c_i$, and so there exists $j \neq i$ such that $c^*_j < c_j$.  But then $f_j(c^*_j) < c_j$ and $f_k(c^*_k) \leq c_k$ for all $k \neq j$, so $A^*$'s total utility is $\sum_{k=1}^n f_k(c^*_k) < \sum_{i=1}^n c_i$, again a contradiction.  Thus, we have $c^*_i = c_i$ for all $i$, and so $\sum_{i \in S_1} c^*_i + \sum_{i \in S_2} c^*_i = \sum_{i=1}^n c_i = 2C$.  So, since $\sum_{i \in S_1} c^*_i \leq C$ and $\sum_{i \in S_2} c^*_i \leq C$, then $\sum_{i \in S_1} c^*_i = \sum_{i \in S_2} c^*_i = C$.  Hence, $\sum_{i \in S_1} c_i = \sum_{i \in S_2} c_i = C$, and $S_1, S_2$ solve the partition problem.
	
	For the only if direction, suppose $S_1, S_2$ are a solution to the partition instance.  Then since $\sum_{i \in S_1} c_i = \sum_{i \in S_2} c_i = C$, we can assign the threads with indices in $S_1$ and $S_2$ to servers 1 and 2, respectively, and get a valid assignment with utility $\sum_{i=1}^n f_i(c_i) = \sum_{i=1}^n c_i$.  This is a maximum utility assignment for $A$, since $f_i(x) \leq c_i$ for all $i$.  Thus, the partition problem reduces to the AA problem, and so the latter is NP-hard for two servers. 
\end{proof}

%% file: algorithm.tex
\section{Approximation Algorithm for Concave Functions}
\label{sec-alg}
\begin{table}\label{table-notation}\small
%\caption{}
\caption{List of notations}
%\caption{$\Delta_{th}$ depending on the relationship between $\bar{p}$ and $k$}
\centering
\begin{tabular}{|c|c|}
\hline
Notations & Definitions\\
\hline
$n$ & Number of all threads \\
$m$ & Number of all servers\\
$S$ & Set of all threads\\
$T$ & Set of all threads\\
$C$ & Each server's resource capacity\\
$f_i$& Original utility function of thread $t_i$\\
$g_i$ & Utility function of thread $t_i$ after linearization\\
$r_i$ & Assignment of thread $t_i$\\
$c_i$ & Resource allocation of thread $t_i$\\
$\hat{c}_i$ & Super-optimal resource allocation of thread $t_i$\\
$E$& Set of unfull threads\\
$D$& Set of full threads\\
$\gamma$ & Maximum super-optimal utility of threads in E\\
$G$ & Total utility at the linear problem\\
$F$ & Total utility at the original problem \\
$F^*$& Optimal total utility at the original problem\\
$\hat{F}$&Super-optimal utility\\
$\alpha$& $2(\sqrt{2}-1)$\\
\hline
\end{tabular}
\end{table}%
%\vspace{-0.2in}
In this section, we present an algorithm for the AA problem when the utility function for each thread is concave.  The algorithm outputs an assignment with total utility at least $\alpha = 2(\sqrt{2}-1) > 0.828$ times the optimal, and runs in $O(mn^2 + n(\log mC)^2)$ time. The algorithm consists of two main steps.  The first step transforms the utility functions, which are arbitrary nondecreasing concave functions and difficult to work with algorithmically, into functions consisting of two linear segments, which are easier to handle. Next, we find an $\alpha$-approximate thread assignment for the linearized functions. We then show that this leads to an $\alpha$ approximate solution for the original concave problem. For ease of exposition, we list the notations used in the remainder of the paper in Table I.

\begin{figure}
	\centering 
		\includegraphics[height=4cm,width=5.2cm]{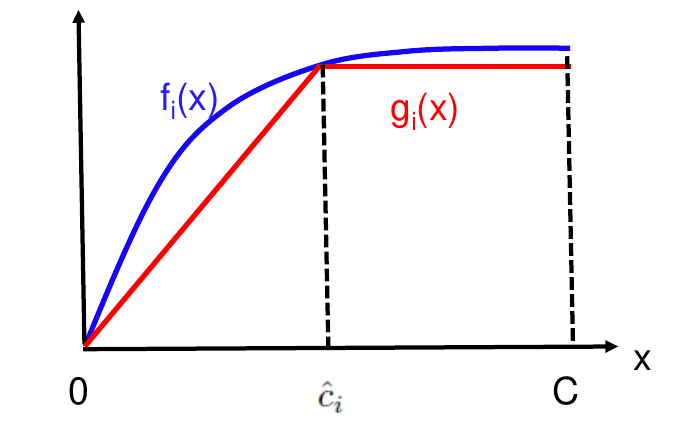}  
	%\hspace{-0.05in}
	\caption{Illustration of linearization of $f_i(x)$ to $g_i(x)$.} 
	\label{figure-linearization}
\end{figure}

\begin{table*}\label{table-example-alg1}
\small
%\caption{}
\caption{Example execution of Algorithm 1}
%\caption{$\Delta_{th}$ depending on the relationship between $\bar{p}$ and $k$}
\centering
\begin{tabular}{|c|c|c|c|c|c|c|}
\hline
Loop index         &$(i,j)$ & Allocation & Assignment & R        & $C_1$ & $C_2$\\
\hline
Loop 1 &(4,1)   & $c_4=6$    & $r_4=1$    &\{1,2, 3\}   &  1     &7 \\
\hline
Loop 2 &(2,2)   & $c_2=3$    & $r_2=2$    &\{1,3\}   &  1     & 3\\
\hline
Loop 3 &(3,2)   & $c_3=3$    & $r_3=2$    &\{1\}   &  1     &1 \\
\hline
Loop 4 &(1,1)   & $c_1=1$    & $r_1=1$    & $\emptyset$  &  0     &1 \\
\hline
\end{tabular}
\end{table*}%

%We show that threads with the original utility functions and the linearized ones have the same maximum total utility. 
\vspace{-0.2cm}
\subsection{Linearization}\label{sec-alg-linearize}
\vspace{-0.1cm}
To describe the linearization procedure, we start with the following definition.
\vspace{-0.1cm}
\blue{
\begin{definition}\label{def-super-opt}
%Given an AA problem $A$ with $m$ servers each with $C$ amount of resources, and $n$ threads with utility functions $f_1, \ldots, f_n$, a \emph{super-optimal allocation} for $A$ is a set of values $\hat{c}_1, \ldots, \hat{c}_n$ that maximizes the quantity $\sum_{i=1}^n f_i(\hat{c}_i)$, subject to $\sum_{i=1}^n \hat{c}_i \leq mC$.  Call $\sum_{i=1}^n f_i(\hat{c}_i)$ the \emph{super-optimal utility} for $A$.
Given an instance $A$ of the AA problem with $m$ servers each with $C$ amount of resources, and $n$ threads with utility functions $f_1, \ldots, f_n$, consider the quantity
\begin{equation*}
\hat{F} = \max_{c_i, i\in [1,n]}\sum_{i=1}^{n}f(c_i)
\end{equation*}
subject to $\sum_{i=1}^n c_i \leq mC$.  Let $\hat{c}_1, \ldots, \hat{c}_n$ be values for $c_1, \ldots, c_n$, respectively, which achieve the optimum $\hat{F}$.  Then we call $\hat{F} = \sum_{i=1}^n f_i(\hat{c}_i)$ the \emph{super-optimal utility} of $A$, and $\hat{c}_1, \ldots, \hat{c}_n$ the \emph{super-optimal allocation} for $A$.
\end{definition}
}

To motivate the above definition, note that for any valid assignment $[(r_1, c_1), \ldots, (r_n, c_n)]$ for $A$, we have $\sum_{i=1} c_i \leq mC$.  Therefore, the utility of the assignment $\sum_{i=1}^n f_i(c_i)$ is at most $\hat{F}$.  Let $F^*$ denote $A$'s maximum utility.  Then we have the following.
\vspace{-0.1cm}
\begin{lemma}
	\label{lem-alg-supopt}
	$F^* \leq \hat{F}$.
\end{lemma}
\vspace{-0.1cm}
Thus, to find an $\alpha$ approximate solution to $A$, it suffices to find an assignment with total utility at least $\alpha \hat{F}$.  We note that the problem of finding $\hat{F}$ and the associated super-optimal allocation can be solved in $O(n (\log mC)^2)$ time using the algorithm from \cite{Galil}, since the $f_i$ functions are concave.  Also, since these functions are nondecreasing, we have the following basic property.
\vspace{-0.1cm}
\begin{lemma}
\label{lem-alg-useall}
$\sum_{i=1}^n \hat{c}_i = mC$.
\end{lemma}
\vspace{-0.1cm}
In the remainder of this section, fix $A$ to be an AA problem consisting of $m$ servers with $C$ resources each and $n$ threads with utility functions $f_1, \ldots, f_n$.   Let $[\hat{c}_1, \ldots, \hat{c}_n]$ be a super-optimal allocation for $A$ computed as in \cite{Galil}.  We define the \emph{linearized} version of $A$ to be another AA problem $B$ with the same set of servers and threads, but where the threads have piecewise linear utility functions $g_1, \ldots, g_n$ defined by
\begin{equation}\label{eqn-g}
g_i(x)=
\begin{cases}
f_i(\hat{c}_i) \frac{x}{\hat{c}_i}  & \text{if } x < \hat{c}_i \\
f_i(\hat{c}_i) & \text{otherwise}
\end{cases}
\end{equation}

Figure \ref{figure-linearization} shows how $f_i(x)$ is linearized to $g_i(x)$. The blue curve is $f_i(x)$ and the red curve is $g_i(x)$.  We first prove the following basic property. 

\vspace{-0.1cm}
\begin{lemma}
	\label{lem-alg-fvsg}
	For any $i \in T$ and $x \in [0, C]$, $f_i(x) \geq g_i(x)$.
\end{lemma}
\begin{proof}
	For $x \in [0, \hat{c}_i]$, we have 
	\begin{eqnarray*}
	f_i(x) &\geq& \frac{\hat{c}_i - x}{\hat{c}_i} f_i(0) + \frac{x}{\hat{c}_i} f_i(\hat{c}_i)\\
	       &\geq& g_i(x),
	\end{eqnarray*}
	where the first inequality follows because $f_i$ is concave, and the second inequality follows because $f_i(0) \geq 0$.  Also, for $x > \hat{c}_i$, $f_i(x) \geq f_i(\hat{c}_i) = g_i(x)$. Hence, the lemma holds.  
\end{proof}
Lemma \ref{lem-alg-fvsg} implies that to find an assignment with total utility at least $\alpha \hat{F}$ at Problem $A$, it suffices to find an assignment with total utility at least $\alpha \hat{F}$ at the linearized problem.

\vspace{-0.1cm}
\subsection{Approximation algorithm for linearized problem}
\label{sec-solve-linear}
\vspace{-0.1cm}
We now describe an $\alpha$ approximation algorithm for the linearized problem.  The pseudocode is given in Algorithm 1.  The algorithm takes as input a super-optimal allocation $[\hat{c}_1, \ldots, \hat{c}_n]$ for $A$ and the resulting linearized utility functions $g_1, \ldots, g_n$, as described in Section \ref{sec-alg-linearize}. Variable $C_j$ represents the amount of resource left on server $j$, and $R$ is the set of unassigned threads. The outer loop of the algorithm runs until all threads in $R$ have been assigned. During each iteration, $U$ is the set of (thread, server) pairs such that the server has at least as much remaining resource as the thread's super-optimal allocation.  If any such pairs exist, then in line 6 we find a thread in $U$ with the greatest utility using its super-optimal allocation.  Otherwise, in line 9 we find a thread that can obtain the greatest utility using the remaining resources of any server. In both cases we assign the thread in line 12 to a server giving it the greatest utility. Lastly, we update the server's remaining resources accordingly.     

\begin{algorithm}\small	
	\begin{algorithmic}[1]
		\caption{}
		\Input Super-optimal allocation $[\hat{c}_1, \ldots,  \hat{c}_n]$, and $g_1, \ldots, g_n$ as defined in Equation \ref{eqn-g} 
		\Statex
		\State $C_j \gets C$ for $j=1, \ldots, m$
		\State $R \gets \{1, \ldots, n\}$
		\While{$R \neq \emptyset$}
		\State $U \gets \{ (i,j) \, | \, (i \in R) \wedge (1 \leq j \leq m) \wedge (\hat{c}_i \leq C_j) \}$
		\If{$U \neq \emptyset$}
		\State $(i,j) \gets \argmax_{(i,j) \in U} \,   g_i(\hat{c}_i)$
		\State $c_i \gets \hat{c}_i$
		\Else
		\State $(i,j) \gets \argmax_{i \in R, 1 \leq j \leq m} \,   g_i(C_j)$
		\State $c_i \gets C_j$
		\EndIf
		\State $r_i \gets j$
		\State $R \gets R - \{i\}$
		\State $C_j \gets C_j - c_i$
		\EndWhile
		\State \textbf{return} $(r_1,c_1),\ldots,(r_n, c_n)$\label{line-final-allocation}
	\end{algorithmic}
\end{algorithm}

\subsubsection{An example execution}
We first give a simple example to illustrate Algorithm 1. Suppose there are 2 servers, 4 threads, and each server has $C = 7$ units of resource.  Let $\hat{c}_1=2, \hat{c}_2=3, \hat{c}_3=3, \hat{c}_4=6$, $f_1(\hat{c}_1)=3, f_2(\hat{c}_2)=6, f_3(\hat{c}_3)=4$ and  $f_4(\hat{c}_4)=7$. Then $g_1(\hat{c}_1)=3, g_2(\hat{c}_2)=6, g_3(\hat{c}_3)=4, g_4(\hat{c}_4)=7$.  Table II shows how Algorithm 1 produces the assignment and allocation for the linearized problem. The algorithm runs for $n=4$ iterations. In each iteration, the column labeled $(i,j)$ shows the (thread, server) pair which increases the overall utility by the largest amount, as in Line 6, 9 of Algorithm 1, column \emph{Allocation} shows the final allocation of the thread in the pair, and column \emph{Assignment} shows the final assignment of the thread in the pair. Initially $C_1=7, C_2=7$, and $R=\{1,2,3,4\}$.  In the first iteration, $U$ consists of all (thread, server) pairs. The pair (4,1) results in the largest utility increase of 7. Thus, $(i,j)=(4,1)$ in line 6 and the algorithm assigns thread $t_4$ to server $s_1$ in line 13.  It allocates 6 units of resource on server $s_1$ to $t_4$ in line 7, sets $R=\{1,2,3\}$, $C_1=1$ and leaves $C_2$ unchanged.  In the second iteration, $U=\{(1,2), (2,2),(3,2)\}$ in line 4, in the third iteration $U=\{(3,2)\}$, and in the fourth iteration $U=\emptyset$, after which the algorithm terminates.

%The pair (4,1) results in the largest utility increase 7. Thus, $(i,j)=(4,1)$ by line 6, Algorithm 1 assigns thread $t_4$ to server $s_1$ by line 13, and allocates 6 amount of resource on server $s_1$ to it by line 7, $R=\{1,2,3\}$, $C_1=1$ by line 14, and $C_2$ remains unchanged.

\vspace{-0.1cm}
\subsection{Analyzing the linearized algorithm} 
\label{sec-alg1-analy}
\vspace{-0.1cm}
We now analyze the quality of the assignment produced by Algorithm 1.  We first define some notations.  Let $D = \{ i \in T \, | \, c_i = \hat{c}_i \}$ be the set of threads whose allocation in Algorithm 1 equals its super-optimal allocation, and let $E = T - D$ be the remaining threads.  We say the threads in $D$ are \emph{full}, and the threads in $E$ are \emph{unfull}. Note that full threads are the ones computed in line 6, and unfull threads are computed in line 9. 

The full threads have the same utility in the super-optimal allocation and the allocation produced by Algorithm 1.  Thus, to show Algorithm 1 achieves a good approximation ratio it suffices to show the utilities of the unfull threads in Algorithm 1 are sufficiently large compared to their utilities in the super-optimal allocation. We first show some basic properties about the unfull threads.  
\vspace{-0.1cm} 
\begin{lemma}
	\label{lem-alg-oneEperserver}
At most one thread from $E$ is assigned to any server.
\end{lemma}
\begin{proof}
\vspace{-0.1cm}
	Suppose for contradiction there are two threads $t_a, t_b$ with $a, b\in E$ assigned to a server $s_k$, and assume that $t_a$ was assigned before $t_b$.  Consider the time of $t_b$'s assignment, and let $S_j$ denote the set of threads assigned to a server $s_j$.  We have $\sum_{i \in S_k} c_i = C$, because $a \in E$, and so $t_a$ was allocated all of $s_k$'s remaining resources in lines 10 and 14 of Algorithm 1.  Also, $\sum_{i \in S_j} c_i = C$ for any $j \neq k$. Indeed, if $\sum_{i \in S_j} c_i < C$ for any $j \neq k$, then $s_j$ has more remaining resources than $s_k$, and so $t_b$ would be assigned to $s_j$ instead of $s_k$ because it can obtain more utility.  Thus, together we have that when $t_b$ is assigned, $\sum_{i \in T} c_i \geq \sum_{j=1}^m \sum_{i \in S_j}c_i = mC$.  Now, $c_i \leq \hat{c}_i$ for all $i \in T$.  Also, since $a, b \in E$, then $c_a < \hat{c}_a$ and $c_b < \hat{c}_b$.  Thus, we have $\sum_{i \in T} \hat{c}_i > \sum_{i \in T} c_i \geq mC$, which is a contradiction because $\sum_{i \in T} \hat{c}_i = mC$ by Lemma \ref{lem-alg-useall}.
\end{proof}
\vspace{-0.1cm}
\begin{lemma}\label{lem-numEthread}
	$|E| \leq m-1$.
\end{lemma}
\vspace{-0.1cm}
\begin{proof}
	Lemma \ref{lem-alg-oneEperserver} implies that $|E| \leq m$, so it suffices to show $|E| \neq m$.  Assume for contradiction $|E| = m$. Then by Lemma \ref{lem-alg-oneEperserver}, for each server $s_k$ there exists a thread $t_a$, $a\in E$ assigned to $s_k$. $t_a$ receives all of $s_k$'s remaining resources, and so $\sum_{i \in S_k} c_i = C$ after its assignment.	Then after all $m$ threads in $E$ have been assigned, we have $\sum_{i \in T} c_i = mC$.  But since $c_a < \hat{c}_a$ for all $a \in E$, and $c_i \leq \hat{c}_i$ for all $i \in T$, we have $\sum_{i \in T} \hat{c}_i > \sum_{i \in T} c_i = mC$, which is a contradiction.  Thus, $|E| \neq m$ and the lemma follows.  
\end{proof}
\vspace{-0.1cm}

The next lemma shows that the total resources allocated to the unfull threads in Algorithm 1 is not too small compared to their super-optimal allocation. We first briefly outline the main proof idea. Note that some servers may have unallocated resources after Algorithm 1 terminates, since there may be some unfull threads which do not get their super-optimal allocations.  Let $D^s$ be the set of \emph{servers} containing only full threads after Algorithm 1 terminates, and let $E^s$ be the remaining servers containing some unfull threads. Recall that $C_j$ represents the unallocated resources in a server $j$.  Then for any thread $i \in E$ and any server $j \in D^s$, we have $c_i\geq  C_j$, because in each iteration of thread assignment we assign a thread to a server giving it the largest utility. Thus, we can show $\frac{\sum_{i\in E} c_i}{\sum_{i\in E} c_i+\sum_{j\in D^s} C_j}\geq \frac{|E|}{|E|+|D^s|}$.  Also, we can show $\sum_{i\in E} c_i+\sum_{j \in D^s} C_j = \sum_{i \in E} \hat{c}_i$. Additionally, the total number of unfull threads and unfull servers (i.e. $|E|+|D^s|$) is equal to the total number of servers $m$, since each unfull thread is distributed to a different server in $E^s$.  From this, we can derive that $\frac{\sum_{i\in E}c_i}{\sum_{i\in E}\hat{c}_i}\geq \frac{|E|}{m}$, as stated in the following lemma.

\vspace{-0.1cm}
\begin{lemma}
	\label{lem-alloc-wrong}
	$\sum_{i\in E}c_i\geq \frac{|E|}{m} \sum_{i\in E}\hat{c}_i$.
\end{lemma}
\begin{proof}
We first partition the servers into sets $U$ and $V$, where $U = \{j \in S \, | \, S_j \subseteq D\}$ is the set of servers containing only full threads, and $V = S-U$ are the servers containing some unfull threads. Let $C_j = C - \sum_{i \in S_j} c_i$ be the amount of unused resources on a server $s_j$ at the end of Algorithm 1.  Then $C_j = 0$ for all $j \in V$, since the unfull thread in $S_j$ was allocated all the remaining resources on $s_j$. So, we have 
\begin{eqnarray*}
\sum_{j \in U} C_j &=& \sum_{j \in S} C_j =  \sum_{j \in S} (C - \sum_{i \in S_j} c_i)\\ 
                   &=& mC - \sum_{i \in T} c_i, 
\end{eqnarray*}
and so 
\begin{equation}\label{eqn-sum-ci}
\sum_{i \in T} c_i = mC - \sum_{i \in U} C_j.
\end{equation}
	
	Next, we have 
\begin{eqnarray*}
	\sum_{i \in T} c_i &=& \sum_{i \in D} \hat{c}_i + \sum_{i \in E} c_i\\ 
	                   &=& mC - \sum_{i \in E} \hat{c}_i + \sum_{i \in E} c_i.
\end{eqnarray*}		
The first equality follows because $c_i = \hat{c}_i$ for $i \in D$, and the second equality follows because $D \cup E = T$ and $\sum_{i \in T} \hat{c}_i = mC$. Combining this with the earlier expression for $\sum_{i \in T} c_i$ in Equation \ref{eqn-sum-ci}, we have 
\begin{equation}
mC - \sum_{i \in U} C_j = mC - \sum_{i \in E} \hat{c}_i + \sum_{i \in E} c_i,
\end{equation} 
and so
	\begin{equation}
	\label{eqn-alloc-wrong1}
	\sum_{i \in E} c_i + \sum_{i \in U} C_i = \sum_{i \in E} \hat{c}_i.
	\end{equation}	
	Now, assume for contradiction that	$\sum_{i\in E}c_i < \frac{|E|}{m} \sum_{i\in E}\hat{c}_i$.  Then by Equation \ref{eqn-alloc-wrong1} we have
	\begin{equation}
	\label{eqn-alloc-wrong2}
	\sum_{i \in U} C_i > \frac{m-|E|}{m} \sum_{i \in E} \hat{c}_i.
	\end{equation}
	We have $|V| = |E|$, since by Lemma \ref{lem-alg-oneEperserver} each server in $V$ contains only one unfull thread.  Thus $|U| = m - |V| = m - |E|$. Using this in Equation \ref{eqn-alloc-wrong2}, we have that there exists an $j \in U$ with \begin{equation}
	\label{eqn-alloc-wrong3}
	C_j \geq \frac{1}{|U|} \sum_{i \in U} C_i > \frac{1}{m} \sum_{i \in E} \hat{c}_i.
	\end{equation}
	We claim that for all $i \in E, j\in U$, $c_i \geq C_j$.  Indeed, suppose  $c_i < C_j$ for some $i$. But since $C_j > c_i$,  $t_i$ should be allocated to $s_j$ because it can obtain greater utility on $s_j$ than its current server, which is a contradiction.  Thus, $c_i \geq C_j$ for all $i \in E$.  Using this and Equation \ref{eqn-alloc-wrong3}, we have $$\sum_{i \in E} c_i \geq \sum_{i \in E} C_j = |E| C_j >  \frac{|E|}{m} \sum_{i \in E} \hat{c}_i$$  However, this contradicts the assumption that $\sum_{i\in E} c_i < \frac{|E|}{m} \sum_{i\in E}\hat{c}_i$.  Thus, the lemma follows.
\end{proof}

\vspace{-0.1cm}
Let $\gamma = \max_{i \in E} g_i(\hat{c}_i)$ be the maximum super-optimal utility of any thread in $E$.  The following lemma says that all of the first $m$ threads assigned by Algorithm 1 are given their super-optimal allocations and have utility at least $\gamma$.
\vspace{-0.1cm}
\begin{lemma}
	\label{lem-Dthruput}
	Let $t_i$ be one of the first $m$ threads assigned by Algorithm 1.  Then $i \in D$ and $g_i(c_i) \geq \gamma$.
\end{lemma}
\vspace{-0.1cm}
\begin{proof}
	To show $i \in D$, note that the $m$ servers all had $C$ resource at the start of Algorithm 1, and fewer than $m$ threads were assigned before $t_i$.  So when $t_i$ was assigned, there was at least one server with $C$ resource.  Then $t_i$ can obtain $\hat{c}_i$ resource on one of these servers, and so $i \in D$.
	
To show  $g_i(c_i) \geq \gamma$, suppose the opposite, and let $j \in E$ be such that $g_j(\hat{c}_j) = \gamma$.  Since $\hat{c}_j \leq C$, and since in $t_i$'s iteration there was some server with $C$ resource, then in that iteration Algorithm 1 would have obtained greater utility by assigning $t_j$ instead of $t_i$, which is a contradiction. Thus, $g_i(c_i) \geq \gamma$.
\end{proof}
\vspace{-0.1cm}
Lemma \ref{lem-Dthruput} implies there are at least $m$ threads in $D$, and so we have the following.
\vspace{-0.1cm}
\begin{corollary}
	\label{cor-Dthruput}
	$\sum_{i \in D} g_i(c_i) \geq m \gamma$.
\end{corollary}
%\vspace{-0.1cm}
The next lemma shows that for the threads in $E$, threads with higher slopes in the nonconstant portion of their utility functions are allocated more resources. 
\vspace{-0.1cm} 
\begin{lemma}
	\label{lem-slope-alloc}
	For any two threads $i,j\in E$, if $\frac{g_i(\hat{c}_i)}{\hat{c}_i} > \frac{g_j(\hat{c}_j)}{\hat{c}_j}$, then $c_i \geq c_j$.
\end{lemma}
\vspace{-0.1cm}
\begin{proof}
Suppose for contradiction $c_i < c_j$, and suppose first that $t_i$ was assigned before $t_j$.  Then when $t_i$ was assigned, there was at least one server with  $c_j$ or more remaining resources.  We have $\hat{c}_i > c_j$, since otherwise $t_i$ can be allocated $\hat{c}_i$ resources, so that $i \not \in E$.  Now, since $\hat{c}_i > c_j > c_i$, then $t_i$ could obtain greater utility by being allocated $c_j$ instead of $c_i$ amount of resources. This is a contradiction.
	
Next, suppose $t_j$ was assigned before $t_i$.  Then when $t_j$ was assigned, there was a server with at least $c_j$ amount of resources.  Again, we have $\hat{c}_i > c_j$.  Indeed, otherwise we have $\hat{c}_i \leq c_j$, and $\hat{c}_j > c_j$ since $j \in E$, and so $t_i$ can be allocated its super-optimal allocation while $t_j$ cannot.  But Algorithm 1 prefers in line 4 to assign threads that can receive their super-optimal allocations, and so it would assign $t_i$ before $t_j$, a contradiction.  Thus, $\hat{c}_i > c_j$.  However, this means that in the iteration in which $t_j$ was assigned, $t_i$ can obtain greater utility than $t_j$, since $g_i(c_j) =  c_j \frac{g_i(\hat{c}_i)}{\hat{c}_i} > c_j \frac{g_j(\hat{c}_j)}{\hat{c}_j}=g_j(c_j)$, where the first equality follows because $\hat{c}_i > c_j$, the inequality follows because  $\frac{g_i(\hat{c}_i)}{\hat{c}_i} > \frac{g_j(\hat{c}_j)}{\hat{c}_j}$, and the second equality follows because $\hat{c}_j > c_j$.  Thus, $t_i$ would be assigned before $t_j$, a contradiction.  The lemma thus follows.
\end{proof}
%\vspace{-0.1cm}
The following facts are used in later parts of the proof.  Facts \ref{fact-abc} and \ref{fact-frac} follow by simple manipulation, while Fact \ref{fact-cs} follows from the Cauchy-Schwarz inequality, and Fact \ref{fact-cheby} is Chebyshev's sum inequality.

\vspace{-0.1cm}
\begin{fact} 
\label{fact-abc}
Given $a, a', b, c > 0$ and $a\geq a'$, $b\leq c$, we have $\frac{a+b}{a+c}\geq\frac{a'+b}{a'+c}$.
\end{fact}
%\begin{proof}
%To prove the fact, it suffices to prove $\frac{a+b}{a+c}-\frac{a'+b}{a'+c}\geq 0$. We have 
%\begin{eqnarray*}
%\frac{a+b}{a+c}-\frac{a'+b}{a'+c}&=&\frac{(a-a')(c-b)}{(a+c)(a'+c)}\\
%                                 &\geq&0.
%\end{eqnarray*}
%The inequality follows since $a\geq a'$ and $c\geq b$. Thus, the fact holds.
%\end{proof}

\begin{fact}\label{fact-frac}
Given $a,a',b,b' > 0$, if $\frac{a}{a'}\leq\frac{b}{b'}$, then $\frac{a}{a'}\leq\frac{a+b}{a'+b'}\leq\frac{b}{b'}$.
\end{fact}
%\begin{proof}
%We first prove $\frac{a}{a'}\leq\frac{a+b}{a'+b'}$. To do this, it suffices to prove $\frac{a}{a'}-\frac{a+b}{a'+b'}\leq 0$. 
%\begin{eqnarray*}
%\frac{a}{a'}-\frac{a+b}{a'+b'}&=&\frac{ab'-a'b}{a'(a'+b')}\\
%                                 &\leq&0.
%\end{eqnarray*}
%The inequality follows because $\frac{a}{a'}\leq\frac{b}{b'}$. Thus, we have $ab'\leq a'b$.
%
%Next, we prove $\frac{a+b}{a'+b'}\leq\frac{b}{b'}$. To do this, it suffices to prove $\frac{a+b}{a'+b'}-\frac{b}{b'}\leq 0$. We have
%\begin{eqnarray*}
%\frac{a+b}{a'+b'}-\frac{b}{b'}&=&\frac{ab'-a'b}{b'(a'+b')}\\
%                                 &<&0.
%\end{eqnarray*}
%The inequality follows because $ab'<a'b$ mentioned earlier. Hence, the fact holds.
%\end{proof}

\begin{fact}\label{fact-cs}
Given $a_1,\ldots, a_n>0$, we have $(\sum_{i=1}^{n} a_i) (\sum_{i=1}^{n} \frac{1}{a_i}) \geq n^2$.
\end{fact}
%\begin{proof} We have
%\begin{eqnarray*} 
%(\sum_{i=1}^{n} a_i) (\sum_{i=1}^{n} \frac{1}{a_i})&=&(\sum_{i=1}^{n} \sqrt{a_i}^2) \big(\sum_{i=1}^{n} \sqrt{\frac{1}{a_i}}^2\big)\\
%                                                 &\geq&\big(\sum_{i=1}^{n}\sqrt{a_i}\sqrt{\frac{1}{a_i}}\big)^2\\
%																								 &=&n^2.
%\end{eqnarray*}
%The first inequality follows from the Cauchy-Schwarz inequality.
%\end{proof}

\begin{fact}\label{fact-cheby}
Given $a_1 \geq a_2 \geq \ldots \geq a_n$ and $b_1 \geq b_2 \geq \ldots \geq b_n$, we have $\sum_{i=1}^n a_i b_i \geq (\frac{1}{n} \sum_{i=1}^n a_i) (\sum_{i=1}^n b_i)$.	
\end{fact}
%\vspace{-0.1cm}
We now state a lower bound on a certain function that will be used in later parts of the proof. 
\vspace{-0.1cm}
\begin{lemma}\label{lem-beta}
	Let $A,d > 0$, and $0 < a_1 \leq a_2 \ldots \leq a_n$. Also, let $\beta=\frac{A+\sum_{i=1}^{n}a_iz_i}{A+\sum_{i=1}^{n}z_i}$, where each $z_i \in [0,d]$. Then 
	\begin{equation*}
	\beta \geq \min_{j = 1, \ldots, n} \left(\frac{A+\sum_{i=1}^{j}a_id}{A+j d}, 1 \right)
	\end{equation*}
\end{lemma}
\begin{proof}
If $a_1 \geq 1$, then Fact \ref{fact-frac} implies that $\beta \geq 1$, and the lemma holds. 	Otherwise, suppose $a_1 < 1$.  Then differentiating $\beta$ with respect to $z_1$, we get
	\begin{equation*}
	\beta'(z_1) = \frac{(a_1-1) A + \sum_{i=2}^{n} (a_1-a_i)z_i} {(A + \sum_{i=1}^{n}z_i)^2}
	\end{equation*}
	Since $a_1\leq a_2\ldots\leq a_n$ and $a_1 < 1$, we have $\beta'(z_1)<0$. Thus, $\beta(z_1)$ is minimized for $z_1 = d$, and we have $$\beta \geq \frac{A+a_1d+\sum_{i=2}^{n}a_iz_i}{A+d+\sum_{i=2}^{j}z_i}$$ 	
	To simplify this expression, suppose first that $\frac{A+a_1d}{A+d} \leq a_2$. Then we have
	$$\beta \geq \frac{A+a_1d+\sum_{i=2}^{n}a_iz_i}{A+d+\sum_{i=2}^{n}z_i} \geq \frac{A+a_1 d}{A+d}.$$
	The second inequality follows because $a_2 \leq \ldots \leq a_n$ and by Fact \ref{fact-frac}.  Thus, the lemma is proved. Otherwise,  $(A+a_1d)/(A+d) > a_2$, and so $$\frac{A+a_1d+\sum_{i=2}^{n}a_iz_i}{A+d+\sum_{i=2}^{n}z_i} \geq \frac{A+a_1d+a_2d+\sum_{i=3}^{n}a_iz_i}{A+2d+\sum_{i=3}^{n}z_i}$$	
	We can simplify the latter expression in a way similar to above, based on whether $(A+a_1 d + a_2 d)/(A+2d) \leq a_3$.  Continuing this way, if we stop at the $j$'th step, then $\beta \geq (A + \sum_{i=1}^j a_i d)/(A+jd)$.  Otherwise, after the $n$'th step, we have $\beta \geq (A + \sum_{i=1}^n a_i d)/(A+nd)$. In either case, the lemma holds.
\end{proof}

%\vspace{-0.1cm}
Since for $i\in E$, $\hat{c}_i > c_i$ and by definition of $g$, we have $g(c_i) = \frac{g_i(\hat{c}_i)}{\hat{c}_i} c_i$ for $i\in E$. Algorithm 1 produces an allocation $c_1, \ldots, c_n$ with total utility $G = \sum_{i\in D} g_i(\hat{c}_i) + \sum_{i\in E} \frac{g_i(\hat{c}_i)}{\hat{c}_i}c_i$.  We now prove that this allocation is an $\alpha$ approximation to the super-optimal utility 
%\vspace{-0.1cm}
$\hat{F} = \sum_{i \in T} f_i(\hat{c}_i)$. 
\begin{lemma}
	\label{lem-alpha-bound}
	$G \geq \alpha \hat{F}$, where $\alpha = 2(\sqrt{2}-1) > 0.828$.
\end{lemma}
%\vspace{-0.1cm}
\begin{proof}
	We have $\hat{F} = \sum_{i \in T} f_i(\hat{c}_i) = \sum_{i \in T} g_i(\hat{c}_i)$ by the definition of the $g_i$. Thus, 
	\begin{eqnarray*}\small
	\frac{G}{\hat{F}} & = & \frac{\sum_{i\in D} g_i(\hat{c}_i)+\sum_{i\in E} \frac{g_i(\hat{c}_i)}{\hat{c}_i} c_i}{\sum_{i\in D} g_i(\hat{c}_i)+ \sum_{i\in E} g_i(\hat{c}_i)} \\
	& \geq & \frac{m\gamma + \sum_{i\in E} \frac{g_i(\hat{c}_i)}{\hat{c}_i} c_i}{m\gamma+\sum_{i\in E} g_i(\hat{c}_i)} \\
	& \geq & \frac{m\gamma+\left(\sum_{i\in E} c_i/|E|\right)\sum_{i \in E} \frac{g_i(\hat{c}_i)}{\hat{c}_i}}{m\gamma+\sum_{i\in E }g_i(\hat{c}_i)} \\
	& \geq & \frac{m\gamma+\left(\sum_{j\in E} \hat{c}_j/m\right)\sum_{i \in E}\frac{g_i(\hat{c}_i)}{\hat{c}_i}}{m\gamma + \sum_{i \in E} g_i(\hat{c}_i)}
	\end{eqnarray*}
	%The first inequality follows because of Corollary \ref{cor-Dthruput} and because $\hat{c}_i \geq c_i$ for all $i$. Thus, we can apply Fact \ref{fact-abc} to the first expression above, letting $\sum_{i\in D} g_i(\hat{c}_i)$ play the role of $a$, $\sum_{i\in E} \frac{g_i(\hat{c}_i)}{\hat{c}_i} c_i$ play the role of $b$, $\sum_{i\in E} g_i(\hat{c}_i)$ play the role of $c$.
	Recall that $\gamma = \max_{i \in E} g_i(\hat{c}_i)$. To obtain the first inequality, we have $\sum_{i\in D} g_i(\hat{c}_i)\geq m\gamma$ by Corollary \ref{cor-Dthruput} and the fact that $\hat{c}_i=c_i$ for all $i\in D$. Also, we have $\sum_{i\in E} \frac{g_i(\hat{c}_i)}{\hat{c}_i} c_i\leq \sum_{i\in E} g_i(\hat{c}_i)$, since $\hat{c}_i \geq c_i$ for all $i\in E$. Thus, we can apply Fact \ref{fact-abc} to the first and second expressions above, letting $\sum_{i\in D} g_i(\hat{c}_i)$ play the role of $a$, $\sum_{i\in E} \frac{g_i(\hat{c}_i)}{\hat{c}_i} c_i$ play the role of $b$, $\sum_{i\in E} g_i(\hat{c}_i)$ play the role of $c$, $m\gamma$ play the role of $a'$. The second inequality follows because by Lemma \ref{lem-slope-alloc}, threads $i \in E$ with larger values of $\frac{g_i(\hat{c}_i)}{\hat{c}_i}$ also have larger values of $c_i$.  Thus, we can apply Fact \ref{fact-cheby} to bring the term $\sum_{i\in E} c_i/|E|$ outside the sum $\sum_{i\in E} \frac{g_i(\hat{c}_i)}{\hat{c}_i} c_i$.  The last inequality follows because of Lemma \ref{lem-alloc-wrong}. Now, assume WLOG that the elements in $E$ are ordered by nonincreasing value of $\hat{c}_i$, so that $\frac{1}{\hat{c}_1} \leq \frac{1}{\hat{c}_2} \leq \ldots \leq \frac{1}{\hat{c}_{|E|}}$.  Let $E_i$ denote the first $i$ elements of $E$ in this order.  For any $i \in E$, we have $g_i(\hat{c}_i) \in [0, \gamma]$.  Thus, applying Lemma \ref{lem-beta} to the last expression above, letting $g_i(\hat{c}_i)$ play the role of $z_i$ and $\frac{\sum_{j\in E} \hat{c}_j}{m \hat{c}_i}$ play the role of $a_i$, and noting $\frac{G}{\hat{F}} \leq 1$, we have
	
	\begin{eqnarray*}\small
	\frac{G}{\hat{F}} & \geq & \min_{i = 1, \ldots, |E|}  \left(\frac{m\gamma + \left(\sum_{j \in E} \hat{c}_j / m \right) \sum_{j \in E_i} \frac{\gamma}{\hat{c}_j}}{m\gamma + i \gamma}	\right)\\
	& \geq & \min_{i = 1, \ldots, |E|}  \left(\frac{m + \frac{1}{m} \left(\sum_{j \in E_i} \hat{c}_j \right) \left( \sum_{j \in E_i} \frac{1}{\hat{c}_j}\right)}{m + i}	\right) \\
	& \geq & \min_{i = 1, \ldots, |E|} \left(\frac{m+\frac{i^2}{m}}{m+i}\right)\\	
	\end{eqnarray*}
	The second inequality follows by simplification and because $\sum_{j\in E} \hat{c}_j \geq \sum_{j\in E_i} \hat{c}_j$ for any $i$. The last inequality follows by Fact \ref{fact-cs}. It remains to lower bound the final expression. Recall $|E| \leq m-1$ by Lemma \ref{lem-numEthread}. Treating $i$ as a real value and taking the derivative with respect to $i$, we find the minimum value is obtained at $i = (\sqrt{2}-1) m$, for which $\frac{G}{\hat{F}} \geq 2(\sqrt{2}-1) = \alpha$.  Thus, the lemma is proved.
\end{proof}
\vspace{-0.3cm}
\subsection{Solving the concave problem}
\vspace{-0.1cm}
To solve the original AA problem with concave utility functions $f_1, \ldots, f_n$, we run Algorithm 1 on the linearized problem to obtain an allocation $c_1, \ldots, c_n$, then simply output this as the solution to the concave problem.  The total utility of this solution is $F = \sum_{i \in T} f_i(c_i)$.  We now show this is an $\alpha$ approximation to the optimal utility $F^*$.  
%\vspace{-0.1cm}
\begin{theorem}
	\label{thm-alg1-ratio}
	$F \geq \alpha F^*$, and Algorithm 1 achieves an $\alpha$ approximation ratio.
\end{theorem}
%\vspace{-0.1cm}
\begin{proof}
	We have $F = \sum_{i \in T} f_i(c_i) \geq \sum_{i \in T} g_i(c_i) \geq \alpha \hat{F} \geq \alpha F^*$, where the first inequality follows because $f_i(c_i) \geq g_i(c_i)$ by Lemma \ref{lem-alg-fvsg}, the second inequality follows by Lemma \ref{lem-alpha-bound}, and the last inequality follows by Lemma \ref{lem-alg-supopt}. Hence, the theorem is proved.
\end{proof}

Next, we give a simple example that shows our analysis of Algorithm 1 is nearly tight.
%\vspace{-0.1cm}
\begin{theorem}
There exists an instance of AA where Algorithm 1 achieves $\frac{5}{6} > 0.833$ times the optimal total utility.
\end{theorem}
%\vspace{-0.1cm}
\begin{proof}
	Consider 3 threads, and 2 servers each with 10 units of resource.  Let
\begin{equation*}
f_1(x) = \begin{cases}
\frac{1}{5}x  & \text{if } x \in [0, 5] \\
1  & \text{if } x > 5.
\end{cases} \quad
\end{equation*}
Also, let $f_2(x) = \frac{1}{10}x$.  Suppose the first two threads both have utility functions $f_1$, and the third thread has utility function $f_2$.  The super-optimal allocation is $[\hat{c}_1, \hat{c}_2, \hat{c}_3] = [5, 5, 10]$.  Algorithm 1 may assign threads 1 and 2 to different servers, with $5$ units of resource each, then assign thread 3 to server 1 with $5$ units of resource.  This has a total utility of $2 \frac{1}{2}$. On the other hand, the optimal assignment is to put threads 1 and 2 on server 1 and thread 3 on server 2.  This has a utility of 3. Thus, the ratio between the total utility achieved by Algorithm 1 and the optimal utility is $\frac{5}{6}>0.833$. Hence, the theorem is proved.
\end{proof}

Lastly, we analyze Algorithm 1's time complexity.
\vspace{-0.1cm}
\begin{theorem}\label{thm-time-1}
	Algorithm 1 runs in $O(mn^2 + n (\log mC)^2)$ time.
\end{theorem}
%The reason behind the theorem is that computing the super-optimal allocation takes $O(n (\log mC)^2)$ time using the algorithm in \cite{Galil}. Then Algorithm 1 runs $n$ loop iterations, where in each iteration it computes the set $U$ with $O(mn)$ elements.
%\vspace{-0.4cm}
\begin{proof}
	Computing the super-optimal allocation takes $O(n (\log mC)^2)$ time using the algorithm in \cite{Galil}. Then the algorithm runs $n$ loop iterations, where in each iteration it computes the set $U$ with $O(mn)$ elements.  Thus, the theorem follows.
\end{proof}

%% file: fast-algorithm.tex
\section{A Faster Algorithm for Concave Functions}
\label{sec-fast-alg}
In this section, we present a faster approximation algorithm for concave utility functions that achieves the same approximation ratio as Algorithm 1 in $O(n(\log mC)^2)$ time.   

\emph{Remark:} Since the faster approximation algorithm has the same approximation ratio as Algorithm 1 and lower time complexity, one may wonder why not omit Algorithm 1 and present the faster algorithm directly. Indeed, Algorithm 1 is the starting point of the faster algorithm, and more natural to understand. By analyzing Algorithm 1, we find some important properties as shown in Section \ref{sec-alg1-analy} that result in the approximation ratio. We propose the faster algorithm by improving Algorithm 1's time complexity while keeping these properties.  

\subsection{Algorithm description}
The pseudocode of the faster approximation algorithm is shown in Algorithm 2. The algorithm also takes as input a super-optimal allocation $\hat{c}_1, \ldots, \hat{c}_n$, which we compute as in Section \ref{sec-alg-linearize}.  It sorts the threads in nonincreasing order of $g_i(\hat{c}_i)$.  It then takes threads $m+1$ to $n$ in this ordering, and sorts them again, this time in nonincreasing order of $g_i(\hat{c}_i) / \hat{c}_i$.  Next, it initializes $C_1, \ldots, C_m$ to $C$, and stores them in a max heap $H$.  $C_j$ represents the amount of remaining resources on server $j$.  The main loop of the algorithm iterates through the threads in order. Each time it chooses the server with the most remaining resources, allocates the minimum of the thread's super-optimal allocation and the server's remaining resources to it, and assigns the thread to the server.  Then $H$ is updated accordingly.   

\begin{algorithm}\small
	
\begin{algorithmic}[1]
\caption{}
\Input Super-optimal allocation $[\hat{c}_1, \ldots,  \hat{c}_n]$, and $g_1, \ldots, g_n$ as defined in Equation \ref{eqn-g} 
\Statex
\State Sort threads in nonincreasing order of $g_i(\hat{c}_i)$ as $t_1, \ldots, t_n$
\State Sort $t_{m+1}, \ldots, t_n$ in nonincreasing order of $g_i(\hat{c}_i) / \hat{c}_i$
\State $C_j \gets C$ for $j=1, \ldots, m$
\State Store $C_1, \ldots, C_m$ in a max-heap $H$
\For{$i = 1, \ldots, n$}
\State $j \gets \argmax_{1 \leq j \leq m} C_j$
%\State find server $s_k$ with $u_k=\max(u_j), j=1,\ldots, m$\label{line-maximum-size}
\State $c_i \gets \min(\hat{c}_i, C_j)$
\State $C_j \gets C_j - c_i$, and update $H$
\State $r_i\gets j$
\EndFor\label{line-assign-2} 
\State \textbf{return} $(r_1,c_1),\ldots,(r_n, c_n)$\label{line-final-allocation}
\end{algorithmic}
\end{algorithm}

\subsubsection{An example execution}
\begin{table}\label{table-example-alg2}
\small
%\caption{}
\caption{Example execution of Algorithm 2}
%\caption{$\Delta_{th}$ depending on the relationship between $\bar{p}$ and $k$}
\centering
\begin{tabular}{|c|c|c|c|c|c|}
\hline
Loop index     &$j$ & Allocation & Assignment & $C_1$ & $C_2$\\
\hline
Loop 1 ($i=1$) &1   & $c_1=6$    & $r_1=1$      &  1     &7 \\
\hline
Loop 2 ($i=2$) &2   & $c_2=3$    & $r_2=2$    &   1     & 4\\
\hline
Loop 3 $(i=3)$ &2   & $c_3=2$    & $r_3=2$    &1   &  2 \\
\hline
Loop 4 $(i=4)$ &2   & $c_4=2$    & $r_4=2$    & 1  &  0  \\
\hline
\end{tabular}
\end{table}%

We show a simple example execution to illustrate Algorithm 2. We consider the same setup as the example for Algorithm 1, namely with $m=2, n=4, C=7$, $\hat{c}_1=2, \hat{c}_2=3, \hat{c}_3=3, \hat{c}_4=6$, $g_1(\hat{c}_1)=3, g_2(\hat{c}_2)=6, g_3(\hat{c}_3)=4, g_4(\hat{c}_4)=7$. Table III shows how Algorithm 2 produces the assignment and allocation of the linearized problem. Initially, the threads are sorted in order $t_4, t_2, t_1, t_3$ in lines 1 and 2.  For presentation purposes, we reindex the threads as $t_1, t_2, t_3, t_4$.  In other words, the newly indexed threads have $\hat{c}_1=6, \hat{c}_2=3, \hat{c}_3=2, \hat{c}_4=3$, $g_1(\hat{c}_1)=7, g_2(\hat{c}_2)=6, g_3(\hat{c}_3)=4, g_4(\hat{c}_4)=7$. $C_1, C_2$ are initially 7, and the algorithm runs for 4 iterations.  As shown in the table III, in each iteration, the column labeled $j$ shows the index of the server with the maximum remaining resource, as in line 6 of Algorithm 2, column \emph{Allocation} shows the final allocation of thread $t_i$, and column \emph{Assignment} shows the final assignment of thread $t_i$. 

In the first iteration, $j=1$ in line 6 since server $s_1$ has the maximum remaining resource. The algorithm assigns thread $t_1$ to server $s_1$ in line 9, and allocates 6 units of resource on server $s_1$ to it in line 7.  $C_1=1$ in line 8, and $C_2$ remains unchanged. The algorithm then continues this way for three more iterations, producing the final allocation and assignment shown in Table III.

%The pair (4,1) results in the largest utility increase 7. Thus, $(i,j)=(4,1)$ by line 6, Algorithm 1 assigns thread $t_4$ to server $s_1$ by line 13, and allocates 6 amount of resource on server $s_1$ to it by line 7, $R=\{1,2,3\}$, $C_1=1$ by line 14, and $C_2$ remains unchanged.

\vspace{-0.2cm}
\subsection{Algorithm analysis}\label{sec-fast-alg-analy}
\vspace{-0.2cm}
We now show Algorithm 2 achieves an $\alpha = 2(\sqrt{2}-1)$ approximation ratio, and runs in $O(n(\log mC)^2)$ time. The proof of the approximation ratio uses exactly the same set of lemmas as in Section \ref{sec-alg-linearize}, \ref{sec-solve-linear} and \ref{sec-alg1-analy}.  The proofs for most of the lemmas are also similar. Rather than replicating them, we will go through the lemmas and point out any differences in the proofs.  Please refer to Sections \ref{sec-alg-linearize}, \ref{sec-solve-linear} and \ref{sec-alg1-analy} for the definitions, lemma statements and original proofs.
\begin{itemize}
%\item \emph{Lemmas \ref{lem-alg-supopt}, \ref{lem-alg-useall}, \ref{lem-alg-fvsg}}  These lemmas deal with the super-optimal allocation, which is the same in Algorithms 1 and 2.
\item \emph{Lemma \ref{lem-alg-oneEperserver}}  The proof of this lemma depended on the fact that in Algorithm 1 if we assign a second $E$ thread to a server, then all the other servers have no remaining resources.  This is also true in Algorithm 2, since in line 6 we assign a thread to a server with the most remaining resources, and so if when we assign a second $E$ thread $t$ to a server $s$ and there was another server $s'$ with positive remaining resources, we would assign $t$ to $s'$ instead, a contradiction.
%\item \emph{Lemma \ref{lem-numEthread}} This follows by exactly the same arguments as the original proof.
\item \emph{Lemma \ref{lem-alloc-wrong}} The only statement we need to check from the original proof is that for all $i \in E$ we have $c_i \geq C_j$. But this is true in Algorithm 2 because if there were any $c_i < C_j$, line 6 of Algorithm 2 would assign thread $i$ to server $j$ instead of $i$'s current server, a contradiction. All the other statements in the original proof then follow.
\item \emph{Lemma \ref{lem-Dthruput}} This follows because lines 1 and 2 of Algorithm 2 show that the first $m$ assigned threads have at least as much super-optimal utility as the remaining $n-m$ threads.  Also, the first $m$ threads must be in $D$, since there is always a server with $C$ resources during the first $m$ iterations of Algorithm 2.  Thus, all threads in $E$ are among the last $n-m$ assigned threads, and their maximum super-optimal utility is no more than the minimum utility of any $D$ thread.
%\item \emph{Corollary \ref{cor-Dthruput}}  This follows immediately from Lemma V.8.
\item \emph{Lemma \ref{lem-slope-alloc}}  As we stated above, all threads in $E$ must be among the last $n-m$ assigned by Algorithm 2.  That is, they are among threads $t_{m+1}, \ldots, t_n$.  In line 2 these threads are sorted in nondecreasing order of $g_i(\hat{c}_i) / \hat{c}_i$.  Thus, the lemma follows. 
%\item \emph{Facts \ref{fact-abc} to \ref{fact-cheby}, Lemma \ref{lem-beta}}  These follow independently of Algorithm 2.
%\item \emph{Lemma \ref{lem-alpha-bound}} The proof of this lemma used only the preceding lemmas, not any properties of Algorithm 1.  Thus, it also holds for Algorithm 2.
\end{itemize}
Given the preceding lemmas, we can state the approximation ratio of Algorithm 2.  The proof of the theorem is the same as the proof of Theorem \ref{thm-alg1-ratio}, and is omitted.
\begin{theorem}
	Let $F$ be the total utility from the assignment produced by Algorithm 2, and let $F^*$ be the optimal total utility.  Then	
	$F \geq \alpha F^*$.
\end{theorem}

Lastly, we analyze Algorithm 2's time complexity.
\begin{theorem}
\label{thm-time-2}
Algorithm 2 runs in $O(n(\log mC)^2)$ time. 
\end{theorem}
\begin{proof}
Finding the super-optimal allocation takes $O(n (\log mC)^2$ time using the algorithm in \cite{Galil}.  Steps 1 and 2 take $O(n \log n)$ time.  Since $C$ is usually large in practice, we can assume that $\log n = O(\log mC)^2$. Each iteration of the main for loop takes $O(\log m)$ time to extract the maximum element from $H$ and update $H$. Thus, the entire for loop takes $O(n \log m)$ time.  Thus, the overall running time is dominated by the time to find a super-optimal allocation, and  the theorem follows.
\end{proof}
Comparing Theorem \ref{thm-time-2} and Theorem \ref{thm-time-1}, we see that Algorithm 2 has lower time complexity than Algorithm 1.

%% file: nonconcave.tex
\section{Algorithm for Nonconcave Functions}\label{sec-nonconcave}
In this section, we present an approximation algorithm for the AA problem for threads with nonconcave utility functions\footnote{\blue{We allow some of the threads to have nonconcave utility functions, and others to have concave ones.}}.  One situation where such functions arise is in cache allocation, where increasing the cache allocated for a thread beyond a certain size allows its entire working set to fit into the fastest level of cache, and leads to a large nonconcave increase in performance. As nonconcave functions pose more challenges than concave ones, our algorithm achieves an approximation ratio of $\frac{1}{2}$ instead of $\alpha=2(\sqrt{2}-1)$. 

The basic model we consider is the same as in Section \ref{sec-model}.  We additionally assume that each utility function consists of one or more concave or convex segments. That is, for each utility function $f:[0, C] \rightarrow \mathbb{Z}^{\geq 0}$, there exist $0 = b_0 < b_1 < b_2 < ... < b_k = C$, such that $f$ is either concave or convex in the interval $(b_{i-1}, b_i]$, for $1 \leq i \leq k$.  Let $s$ be the maximum number of segments in any utility function.  Note that in the case of $s = C$, $f$ can be an arbitrary nondecreasing function.  However, it was observed in \cite{Lai} that $s$ is typically a small value, \emph{e.g.} 3 in practice.  Moreover, a small value of $s$ enables our algorithm to run faster.  

We present an approximation algorithm, which we call \blue{\emph{AANC}}, to solve the AA problem in this setting. \emph{AANC} follows a similar structure as Algorithm 2, with the following difference.  Recall that the input to Algorithm 2 is a super-optimal allocation, which for concave utility functions can be computed by the algorithm from  \cite{Galil}.  In the nonconcave setting, we use the fast algorithm from \cite{Lai}, which computes an optimal allocation $[\hat{c_1}, \ldots, \hat{c_n}]$ for a single server with $mC$ resources in $O(s nmC \alpha(mC) (\log mC)^2)$ time, where $\alpha$ denotes the inverse Ackermann function. Thus, we define \emph{AANC} to be the same as Algorithm 2, but using $[\hat{c_1}, \ldots, \hat{c_n}]$ as the initial input.   

\emph{AANC} achieves an approximation ratio of $\frac{1}{2}$. The main reason for the reduced approximation ratio is that the property that $f_i\geq g_i$ for all $i\in T$ used in Lemma \ref{lem-alg-fvsg} does not hold for nonconcave utility functions. Thus, while Lemma \ref{lem-alpha-bound}, which states that the total utility using the linearized $g_i$ functions is at least $\alpha$ times the super-optimal utility, continues to hold in the nonconcave case, it does imply that the total utility under the original nonconcave utility functions $f_i$ is also $\geq \alpha$ times the super-optimal utility.  Nevertheless, we can show that even in the nonconcave setting, the number of threads which do not get their super-optimal allocation is small, and their super-optimal utilities are smaller than those of threads which do get their super-optimal allocations. Thus, the threads still obtain a large fraction of the super-optimal utility, leading to a $\frac{1}{2}$ approximation ratio.  

We now analyze the quality of the assignment produced by \emph{AANC}. The proof uses several results, including Lemma \ref{lem-alg-supopt}, Lemma \ref{lem-numEthread} and Corollary \ref{cor-Dthruput}, which were proven in Sections \ref{sec-alg-linearize} and \ref{sec-alg1-analy}.  The proofs of these lemmas did not rely on the concavity of the $f_i$ functions, and thus continue to hold in the nonconcave setting. Please refer to Sections \ref{sec-alg-linearize}, \ref{sec-solve-linear} and \ref{sec-alg1-analy}, \ref{sec-fast-alg} for the definitions, lemma statements and original proofs.

We first show that the super-optimal utility of the full threads in $D$ is no smaller than that of the unfull threads in $E$. 
\begin{lemma}\small\label{lem-nonconcave-sumDsumE}
$\sum_{i\in D}f_i(\hat{c}_i)\geq \sum_{i\in E}f_i(\hat{c}_i)$.
\end{lemma}
\begin{proof}
Recall the definition of $\gamma = \max_{i \in E} g_i(\hat{c}_i)$ in Section \ref{sec-alg1-analy}. We first claim  $\sum_{i \in D} f_i(\hat{c}_i) \geq m \gamma$.  This is because by the definition of function $g_i$, we have $f_i(\hat{c}_i)=g_i(\hat{c}_i)=g_i(c_i)$ for any thread $i\in D$. Thus, by Corollary \ref{cor-Dthruput}, we have $\sum_{i \in D} f_i(\hat{c}_i) \geq m \gamma$. 

Next, we claim $\sum_{i \in E} f_i(\hat{c}_i) \leq m \gamma$.  We have $\forall i\in E, g_i(\hat{c}_i)\leq \gamma$. Also, $|E| \leq m-1$ by Lemma \ref{lem-numEthread}.  Thus, since $f_i(\hat{c}_i)=g_i(\hat{c}_i)$, we have $\sum_{i \in E} f_i(\hat{c}_i) = \sum_{i \in E} g_i(\hat{c}_i) \leq (m-1) \gamma$.

Combining the above, we have $\sum_{i \in D} f_i(\hat{c}_i) \geq \sum_{i \in E} f_i(\hat{c}_i)$.
\end{proof}

\emph{AANC} produces an allocation $c_1, \ldots, c_n$ with total utility $F =\sum_{i \in D} f_i(\hat{c}_i)+\sum_{i \in E} f_i(c_i)$.  We now prove that this allocation is an $\frac{1}{2}$ approximation to the super-optimal utility $\hat{F} =\sum_{i \in D} f_i(\hat{c}_i)+\sum_{i \in E} f_i(\hat{c}_i)$.
\begin{lemma}\small\label{lem-nonconcave-bound}
$F\geq\frac{1}{2}\hat{F}$.
\end{lemma}
\begin{proof}
We have
	\begin{eqnarray*}\small
	\frac{F}{\hat{F}} & = & \frac{\sum_{i\in D} f_i(\hat{c}_i)+\sum_{i\in E}f_i(c_i)}{\sum_{i\in D} f_i(\hat{c}_i)+ \sum_{i\in E} f_i(\hat{c}_i)} \\
                    &\geq&\frac{\sum_{i\in D} f_i(\hat{c}_i)}{\sum_{i\in D} f_i(\hat{c}_i)+ \sum_{i\in E} f_i(\hat{c}_i)}\\
										&\geq&\frac{1}{2}.
	\end{eqnarray*}
The last inequality follows since $\sum_{i\in D}f_i(\hat{c}_i)\geq \sum_{i\in E}f_i(\hat{c}_i)$ by Lemma \ref{lem-nonconcave-sumDsumE}. Thus, the lemma is proved.
\end{proof}

Recall that $F^*$ is the optimal total utility. Also, $F^* \leq \hat{F}$, by Lemma \ref{lem-alg-supopt}.  Thus, combining these with Lemma \ref{lem-nonconcave-bound}, we have the following bound on the approximation ratio of \emph{AANC}.

\begin{theorem}
	$F \geq \frac{1}{2} F^*$.
\end{theorem}
	
Next, we give an instance of AA which shows our analysis of \emph{AANC} is nearly tight.
\vspace{-0.1cm}
\begin{theorem}
For any $\epsilon > 0$, there exists an instance of AA such that \emph{AANC} achieves $\frac{1}{2}+\epsilon$ times the optimal total utility.
\end{theorem}
%\vspace{-0.1cm}
\begin{proof}
	Consider $2m-1$ threads, and $m$ servers each with $m$ units of resource.  Let
\begin{equation*}\small
f_1(x) = \begin{cases}
x  & \text{if } x \in [0, 1] \\
1  & \text{if } x > 1,
\end{cases} \quad
\end{equation*}
and
\begin{equation*}\small
f_2(x) = \begin{cases}
0  & \text{if } x \in [0, m-1] \\
1  & \text{if } x > m-1.
\end{cases} \quad
\end{equation*}
 Suppose the first $m$ threads all have utility functions $f_1$, and the remaining $m-1$ threads have utility functions $f_2$. The super-optimal allocation is $\hat{c}_1=\ldots=\hat{c}_m=1$, $\hat{c}_{m+1}=\ldots=\hat{c}_{2m-1}=m$.  \emph{AANC} may assign threads $1, 2,\ldots, m$ all to different servers, with $1$ unit of resource each, and then assign threads $m+1,\ldots, 2m-1$ also all to different servers with $m-1$ units of resource each.  This achieves a total utility of $m$.  On the other hand, the optimal assignment is to put threads $1, 2,\ldots, m$ on server 1 and threads $m+1,\ldots, 2m-1$ on servers $2,3,\ldots,m$, respectively.  This has a total utility of $2m-1$.  Thus, \emph{AANC} achieves an approximation ratio of $\frac{m}{2m-1} \in [\frac{1}{2}, \frac{1}{2}+\epsilon]$ for a sufficiently large $m$.  
\end{proof}

Lastly, we analyze \emph{AANC}'s time complexity.
\begin{theorem}
\label{thm-time-3}
\emph{AANC} runs in $O(s nmC \alpha(mC) (\log mC)^2)$ time, where $s$ is the maximum number of concave or convex segments in any utility function $f_i$, and $\alpha(mC)$ is the inverse Ackermann function. 
\end{theorem}
\begin{proof}
Finding the super-optimal allocation takes $O(s nmC \alpha(mC) (\log mC)^2)$ time using the algorithm in \cite{Lai}.   Note that $\alpha(mC) \leq 4$ for all realistic values of $m$ and $C$.  Similar to the proof of Theorem \ref{thm-time-2}, the overall running time of \emph{AANC} is dominated by the time to find a super-optimal allocation, and the theorem follows.
\end{proof}

%% file: experiment.tex
\section{Experimental Evaluation}
\vspace{-0.1cm}
In this section we experimentally evaluate the performance of our algorithms using both synthetic and real-world utility functions.  We compare the total utility our algorithms achieve with the super-optimal (SO) utility, which is at least as large as the optimal utility.  We also compare the algorithm with several simple but practical heuristics we name UU, UR, RU and RR\footnote{\blue{To the best of our knowledge, we are the first to study the thread assignment and resource allocation problems in a unified context, and we are unaware of other algorithms in the literature which can be directly compared to our algorithms. Thus, to evaluate our algorithms' performance, we compare them to several simple but practically useful heuristics.}}. The UU (uniform-uniform) heuristic assigns threads in a round robin manner to the servers, and allocates the threads assigned to a server the same amount of resources. UR (uniform-random) assigns threads in a round robin manner, and allocates threads a random amount of resources on each server. RU (random-uniform) assigns threads to random servers, and equally allocates resources on each server. Finally, RR (random-random) randomly assigns threads and allocates them random amounts of resource.  

Our simulation experiments use threads with synthetic random utility functions generated according to various probability distributions as described below.   

To generate the random concave utility functions, we fix an amount of resource $C$ on each server, and set the value of the utility function at 0 to be 0.  We generate two values $v$ and $w$ according to the distribution $H$, conditioned on $w \leq v$, and set the value of the utility function at $\frac{C}{2}$ to $v$, and the value at $C$ to $v+w$. Then we apply the PCHIP interpolation function from Matlab to the three generated points to produce a concave utility function. 

To generate the random nonconcave utility functions, we also set the function to be 0 at allocation 0.  We generate a value $b\in (0,C)$ according to a distribution $H$, then divide  $[0, C]$ into two segments $[0,b]$ and $[b,C]$.  Then we generate two concave functions in these segments using a similar approach as for generating a concave utility function. In all the experiments on synthetic utility functions, we set the number of servers to be $m=8$ and the resource size to be $C=100$, and test the effects of varying different parameters. One parameter is $\beta = \frac{n}{m}$, the average number of threads per server. The results in the following sections show the average performance from 1000 random trials.

%We also use real world utility functions derived from 10 SPEC CPU benchmarks: \emph{fma3d}, \emph{gzip}, \emph{gap}, \emph{crafty}, \emph{parser}, \emph{applu}, \emph{twolf}, \emph{mcf}, \emph{apsi}, \emph{swim}. Each benchmark's execution speed (utility) is measured using its IPC (instructions per cycle) as a function of the amount of L2 cache the program is allocated\footnote{In multicore processors with multiple levels of cache, the first level L1 cache is typically privately owned by each core, while the L2 cache is shared by different cores.}. The first 8 benchmarks have concave utility functions and the last 2 benchmarks have nonconcave utility functions \cite{Qureshi, Lai}. In the experiment on real world utility functions, we set the number of servers to be $m=2$, and the number of threads to be $n=8$, representing some 8 among the 10 benchmarks described earlier. We then test the effects of varying the server resource size $C$, representing the size of the L2 cache.  We implemented the algorithms in Matlab, and ran them on a desktop PC with an 8 core Intel CPU and 16 GB of memory.  

\vspace{-0.2cm}
\subsection{Concave Utility Functions}
\vspace{-0.1cm}
We first look at the performance of Algorithm 2 with concave utility functions; we omit testing Algorithm 1, since it achieves the same approximation ratio as Algorithm 2.  \blue{We note that for $m=8, n=100, C=100$, Algorithm 2 terminated in 9 ms.  This can be further improved by implementing the algorithm in C instead of Matlab and using better data structures.}
%\vspace{-0.1cm}
\subsubsection{Uniform and normal distributions}
%\begin{figure}
%	\centering 
%	\subfigure[Uniform distribution]{ 
%		\label{figure-uniform-nvsm} %% label for first subfigure 
%		\includegraphics[height=3.25cm,width=3.75cm]{result/figure_uniform_nvsm.pdf}} 
%	\hspace{-0.05in} 
%	\subfigure[Normal distributionn]{ 
%		\label{figure-normal-nvsm} %% label for second subfigure 
%		\includegraphics[height=3.25cm,width=3.75cm]{result/figure_normal_nvsm.pdf}}  
%	\hspace{-0.05in}
%	\caption{Performance of Algorithm 2 versus SO,UU,RU,UR, and RR as a function of $\beta$ under the uniform and normal distributions.} 
%	\label{figure-nvsm}
%\end{figure}

\begin{figure}
\label{figure-uniform-normal}
	\centering 
	\subfigure[Uniform distribution]{ 
		\label{figure-uniform-nvsm-average} %% label for first subfigure 
		\includegraphics[height=4cm,width=5.2cm]{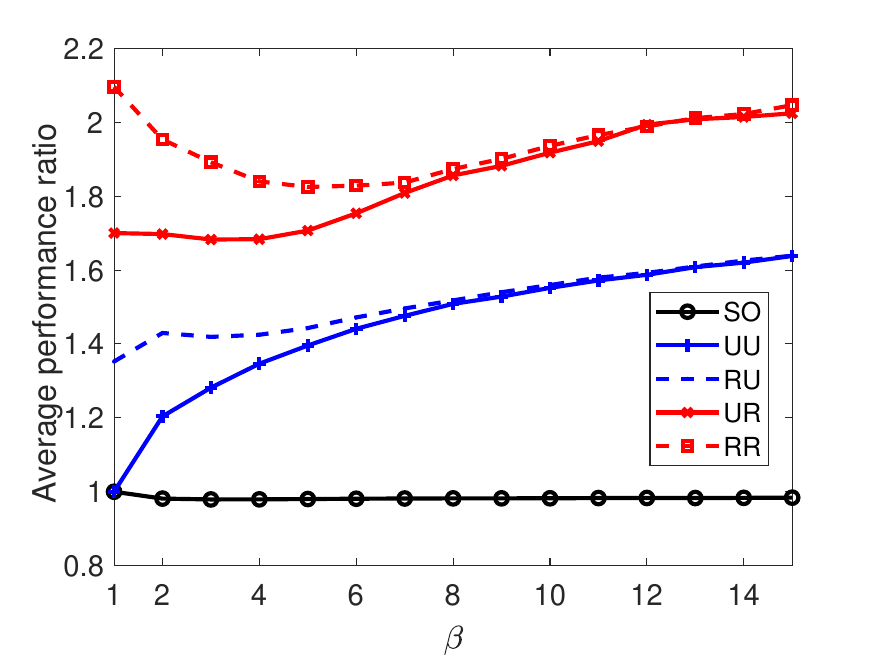}} 
	%\hspace{-0.05in} 
	\subfigure[Normal distribution]{ 
		\label{figure-normal-nvsm-average} %% label for second subfigure 
		\includegraphics[height=4cm,width=5.2cm]{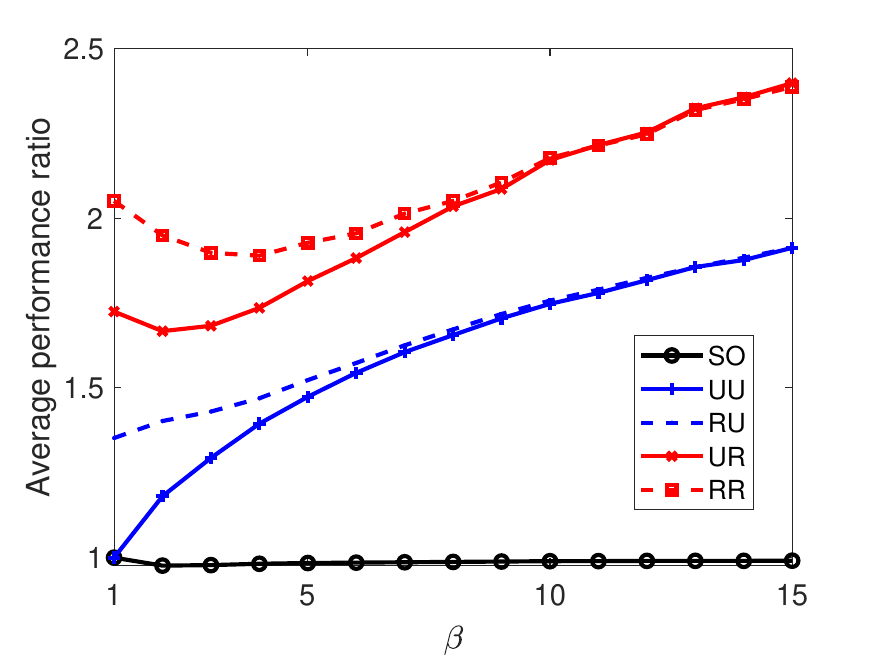}}  
	%\hspace{-0.05in}
	\caption{Average performance of Algorithm 2 versus SO,UU,RU,UR, and RR as a function of $\beta$ under the uniform and normal distributions.} 
	\label{figure-nvsm-average-average}
\end{figure}

%\begin{figure}
%	\centering 
%	\subfigure[Uniform distribution]{ 
%		\label{figure-uniform-nvsm-max} %% label for first subfigure 
%		\includegraphics[height=5cm,width=6.5cm]{result/figure_uniform_nvsm_max.pdf}} 
	%\hspace{-0.05in} 
%	\subfigure[Normal distribution]{ 
%		\label{figure-normal-nvsm-max} %% label for second subfigure 
%		\includegraphics[height=5cm,width=6.5cm]{result/figure_normal_nvsm_max.pdf}}  
	%\hspace{-0.05in}
%	\caption{Maximum performance of Algorithm 2 versus SO,UU,RU,UR, and RR as a function of $\beta$ under the uniform and normal distributions.} 
%	\label{figure-nvsm-average-average}
%\end{figure}

We first look the total utility obtained by Algorithm 2 compared to those obtained by the SO, UU, UR, RU and RR algorithms on threads with concave utility functions generated according to the uniform and normal distributions. We set the mean and standard deviation of the normal distribution to be $10$ and $20$, respectively. 

Figures \ref{figure-uniform-nvsm-average} and \ref{figure-normal-nvsm-average} show the average over 1000 random runs of the ratio of Algorithm 2's total utility to the utilities of the other algorithms, for $\beta$ varying between 1 to 15. The behaviors for both distributions are similar.  Compared to SO, Algorithm 2's utility ratio never drops below 0.99, showing that it always achieves at least 99\% of the optimal utility.  The ratios of Algorithm 2's total utility compared to those of UU, UR, RU and RR are always above 1, showing that it always perform better than the simple heuristics.  For small values of $\beta$, UU performs well.  Indeed, for $\beta = 1$, UU achieves the optimal utility because it places one thread on each server and allocates it all the resources.  UR does not achieve optimal utility even for $\beta=1$, since it allocates threads random amounts of resources. RU and RR may allocate multiple threads per server, and also do not achieve the optimal utility.  As $\beta$ grows, the performance of the heuristics gets worse relative to Algorithm 2.  This is because as the number of threads grows, it becomes more likely that some threads have very high maximum utility.  These threads need to be assigned and allocated carefully.  For example, they should be assigned to different servers and allocated as much resources as possible.  The heuristics likely fail to do this, and hence obtain low performance. The performance of UR and RR, as well as those of UU and RU converge as $\beta$ grows.  This is because both random and uniform assignments assign the threads roughly evenly between the servers for large $\beta$.  Also, the performance of UU and RU are substantially better than UR and RR, which indicates that the way in which resources are allocated has a larger effect on performance than how threads are assigned, and that uniform allocation is generally better than random allocation.  

%Figures \ref{figure-uniform-nvsm-max} and \ref{figure-normal-nvsm-max} show the maximum ratio of Algorithm 2's total utility versus the utilities of the other algorithms, for $\beta$ varying between 1 to 15. Here we see the same trends as those under the average ratio, namely that Algorithm 2 always performs very close to optimal, while the performance of the heuristics gets worse with increasing $\beta$. However, the rate of performance degradation is slower than that under the average performance. Also, we note that there is some fluctuation in the value of maximum performance ratio. We believe this may be due to not running enough random trials, since the maximum performance ratio is more sensitive to the sample size than average performance ratio. In addition, our algorithm is much better than other algorithms when $\beta$ is large. For example, for $\beta=15$, under the uniform distribution, our algorithm is 3.8 times better than UR and RR, and 2 times better than UU and RU. Under the normal distribution, our algorithm is 4.8 times better than UR and RR, and 2.5 times better than UU and RU.  

%We also look at the performance of Algorithm 2 using threads with utility functions generated according to the power law distribution. We see the same trends as those under the uniform and normal distributions, namely that Algorithm 2 always performs very close to optimal, while the performance of the heuristics gets worse with increasing $\beta$. Due to space limitations, the details are given in our online technical report \cite{PanAssignReport}.
%\vspace{-0.1cm}
\subsubsection{Power law distribution}

\begin{figure}
	\centering 
	%\hspace{-0.2in}
	\subfigure[Varying $\beta$]{ 
		\label{figure-powerlaw-nvsm-average} 
		\includegraphics[height=4cm,width=5.2cm]{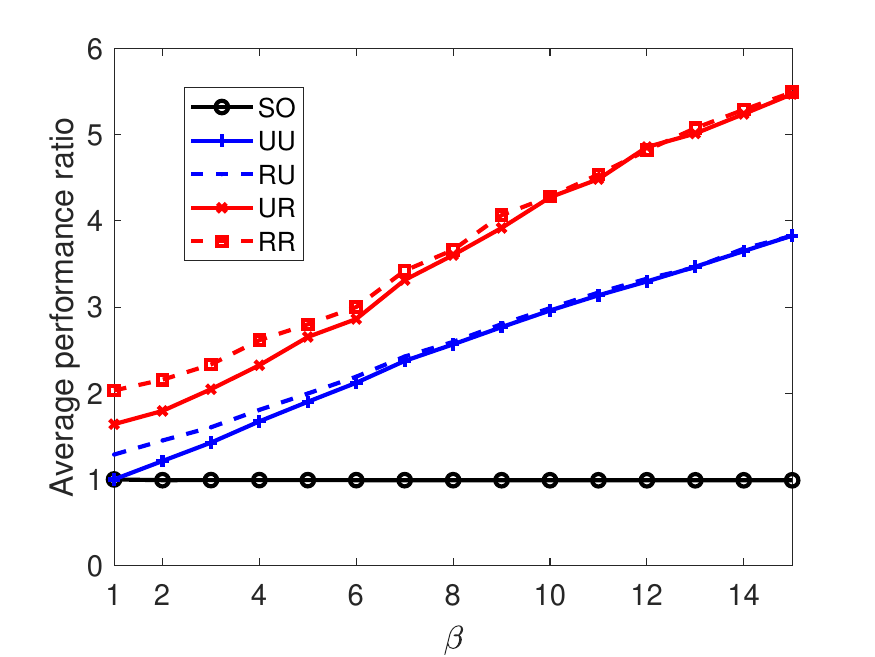}} 
	%\hspace{-0.05in} 
	\subfigure[Varying $\mu$]{ 
		\label{figure-powerlaw-gamma-average} 
		\includegraphics[height=4cm,width=5.2cm]{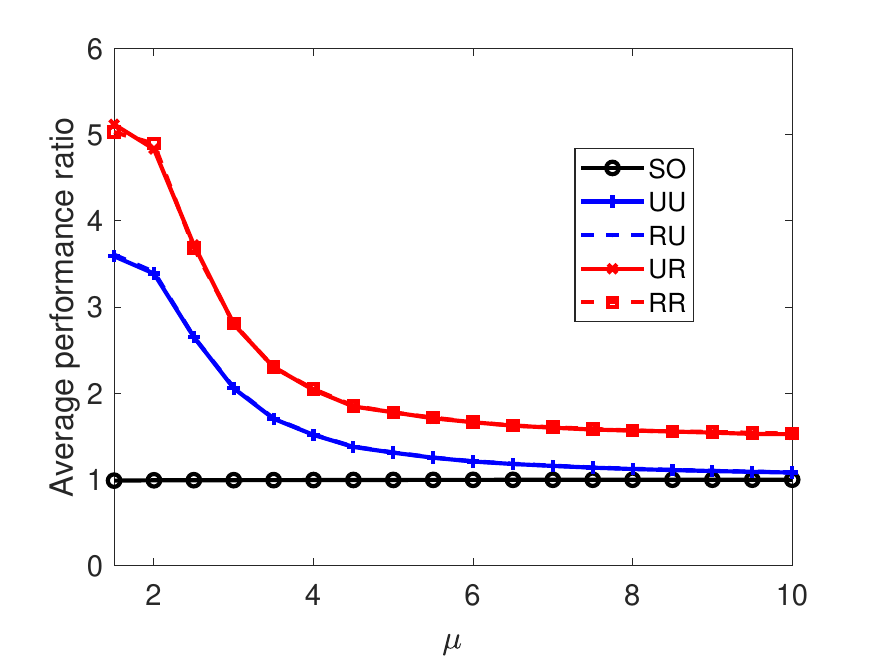}}    
	\caption{Average performance of Algorithm 2 versus SO,UU,RU,UR, and RR as a function of $\beta$ and $\mu$ under the power law distribution.} 
	%\vspace{-0.2in}
	\label{figure-powerlaw-distribution-average}
\end{figure}

\begin{figure}
	\centering 
	%\hspace{-0.2in}
	\subfigure[Varying $\beta$]{
		\label{figure-discrete-nvsm-average} %% label for first subfigure 
		\includegraphics[height=4cm,width=5.2cm]{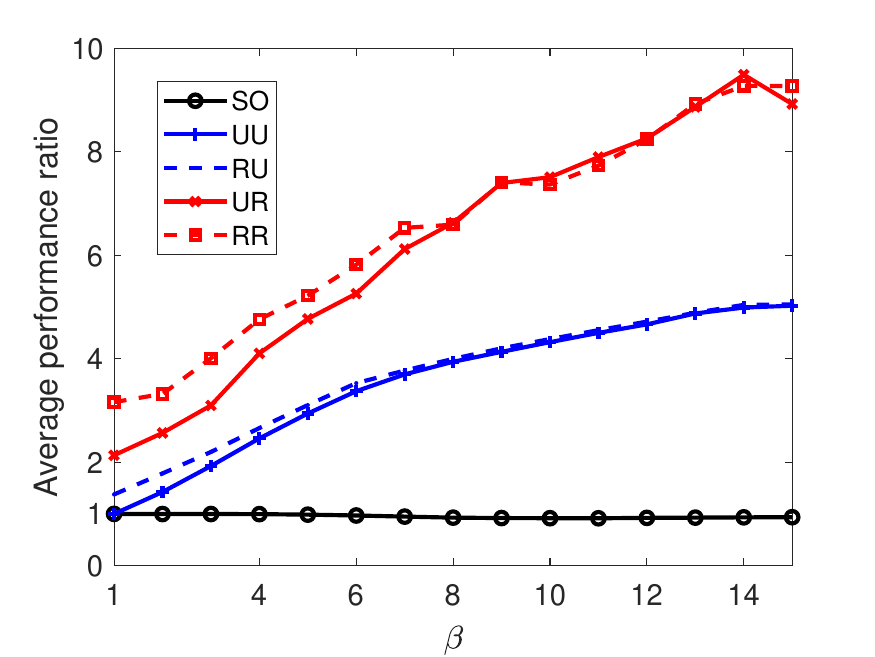}}  
	\subfigure[Varying $\gamma$]{ 
		\label{figure-discrete-gamma-average} %% label for second subfigure 
		\includegraphics[height=4cm,width=5.2cm]{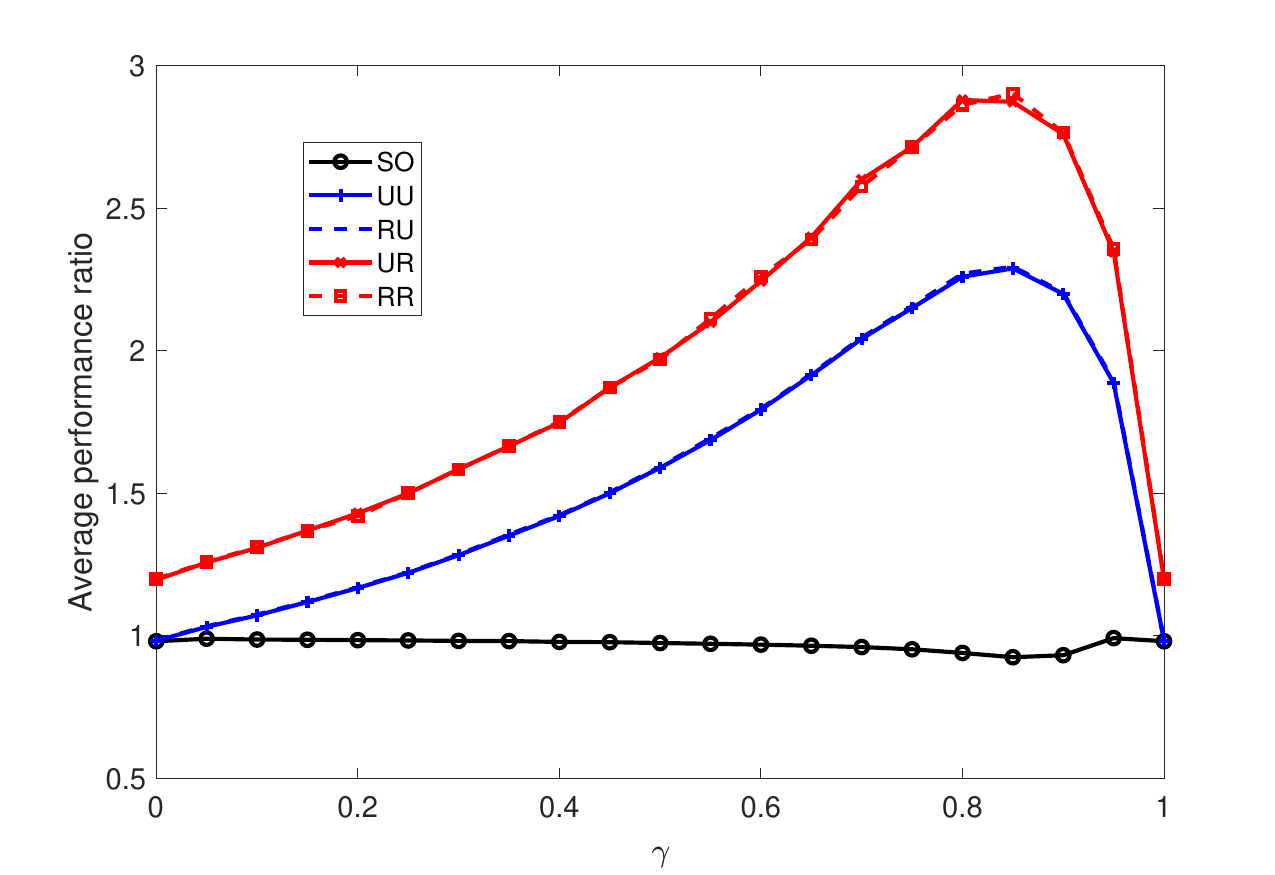}}  
	\subfigure[Varying $\theta$]{
		\label{figure-discrete-theta-average} %% label for first subfigure 
		\includegraphics[height=4cm,width=5.2cm]{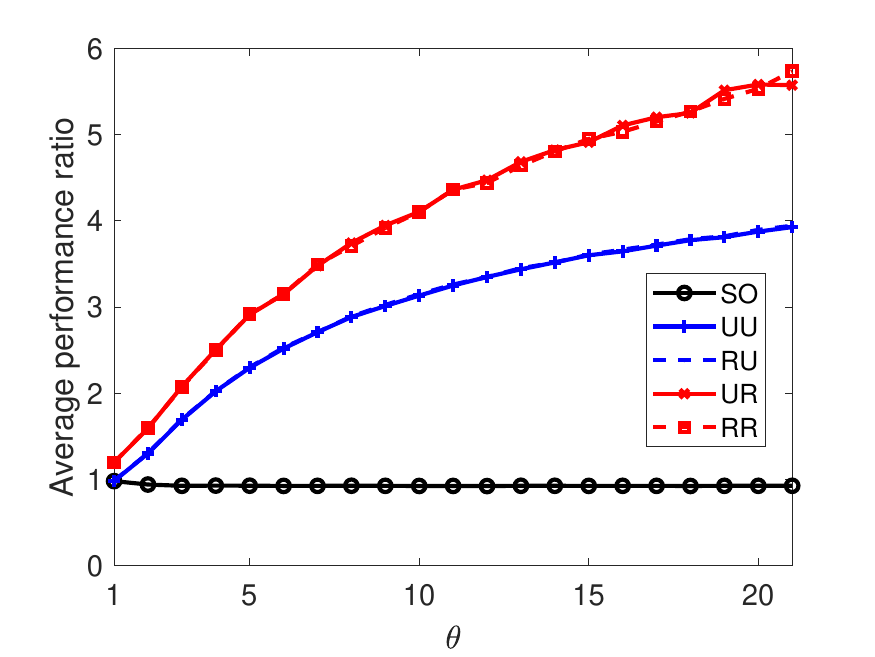}}  
	\caption{Average performance of Algorithm 2 versus SO,UU,RU,UR, and RR as a function of $\beta, \gamma$ and $\theta$ under the discrete distribution.} 
	%\vspace{-0.2in}
	\label{figure-discrete-distribution-average}
\end{figure}

We now look at the performance of Algorithm 2 using threads with utility functions generated according to the power law distribution.  Here, each value $x$ has a probability $\lambda x^{-\mu}$ of occurring, for some $\mu > 1$ and normalization factor $\lambda$.  Figure \ref{figure-powerlaw-nvsm-average} shows the effect of varying $\beta$ while fixing $\mu = 2$.  Here we see the same trends as those under the uniform and normal distributions, namely that Algorithm 2 always performs very close to optimal, while the performance of the heuristics gets worse with increasing $\beta$.  However, the rate of performance degradation is faster than with the uniform and normal distributions.  This is because the power law distribution with $\mu = 2$ is more likely to generate threads with very different maximum utilities.  These threads must be carefully assigned and allocated, which the heuristics fail to do.  For $\beta=15$, Algorithm 2 is 3.9 times better than UU and RU, and 5.7 times better than UR and RR.

Figure \ref{figure-powerlaw-gamma-average} shows the effect of varying $\mu$, using a fixed $\beta = 5$.  Algorithm 2's performance is nearly optimal.  In addition, the performance of the heuristics improves as $\mu$ increases.  This is because for higher values of $\mu$, it is unlikely that there are threads with very high maximum utilities.  So, since the maximum utilities of the threads are roughly the same, almost any even assignment of the threads works well.  Despite this, we still observe that UU and RU perform better than UR and RR.  This is because when the threads are roughly the same, the concavity of the utility functions implies the optimal allocation is to give each thread nearly the same amount of resources.  This is done by UU and RU but not by UR and RR.
%\vspace{-0.1cm}
\subsubsection{Discrete distribution}
We now look at the performance using utility functions generated by a discrete distribution. This distribution takes on only two values $\ell, h$, with $\ell < h$. $\gamma$ is a parameter that controls the probability that $\ell$ occurs, and $\theta = \frac{h}{\ell}$ is a parameter that controls the relative size of the values. Figure \ref{figure-discrete-nvsm-average} shows Algorithm 2's performance as we vary $\beta$, fixing $\gamma = 0.85$ and $\theta = 100$.  The same trends as with the other distributions are observed. Specifically, our algorithm is much better than other algorithms when $\beta$ is large. For example, for $\beta=15$, our algorithm is 9 times better than UR and RR, and 5 times better than UU and RU. The reason behind the significant performance improvement is that the discrete distribution under $\theta=100$ is very likely to generate threads with very different maximum utilities.  Figure \ref{figure-discrete-gamma-average} shows the effect of varying $\gamma$, when $\beta = 5$ and $\theta = 5$.  Our algorithm achieves the lowest performance for $\gamma = 0.8$, when we achieve 92\% of the super-optimal utility. The four heuristics also perform worst for this value.  For $\gamma$ close to 0 or 1, all the heuristics perform well, since these correspond to instances where either $h$ or $\ell$ is very likely to occur, so that almost all the threads have the same maximum utility. Lastly, we consider the effect of varying $\theta$. Here, as $\theta$ increases, the difference between the high and low utilities becomes more evident, and the effects of poor thread assignments or misallocating resources become more serious.  Hence, the performance of the heuristics decreases with $\theta$.  Meanwhile, Algorithm 2 always achieves over 92\% of the optimal utility.

\vspace{-0.2cm}
\subsection{Nonconcave Utility Functions}
\vspace{-0.1cm}
In this section we evaluate the performance of \emph{AANC} on threads with nonconcave utility functions. \blue{We note that for $m=8, n=100, C=100$, \emph{AANC} terminates in 9.47 seconds.  Our implementation of AANC again used Matlab, and we believe the running time of the algorithm can be substantially improved using more optimized code.  However, we leave a more efficient implementation as future work.} 
%\vspace{-0.1cm}
\subsubsection{Uniform distribution}
We first consider the total utility obtained by \emph{AANC} compared to SO, UU, UR, RU and RR on threads with nonconcave utility functions generated according to the uniform distribution. Specifically, for each thread we generate a nonconcave utility function with two concave segments, where each segment is generated using the approach described earlier in this section. Figures \ref{figure-uniform-nvsm-average-nonconcave} shows the average ratio of \emph{AANC}'s total utility versus the utilities of the other algorithms, for $\beta$ varying between 1 to 15. The behaviors are similar to that for Algorithm 2 using concave utility functions generated according to the uniform distribution. Compared to SO, our performance never drops below 0.98, which is slightly lower than the average ratio of 0.99 achieved by Algorithm 2. Likewise, \emph{AANC} performs similarly to Algorithm 2 for other distributions.  Due to space limitations we omit an in-depth discussion.

\begin{figure}
	\centering 
	%\hspace{-0.2in}  
	\includegraphics[height=4cm,width=5.2cm]{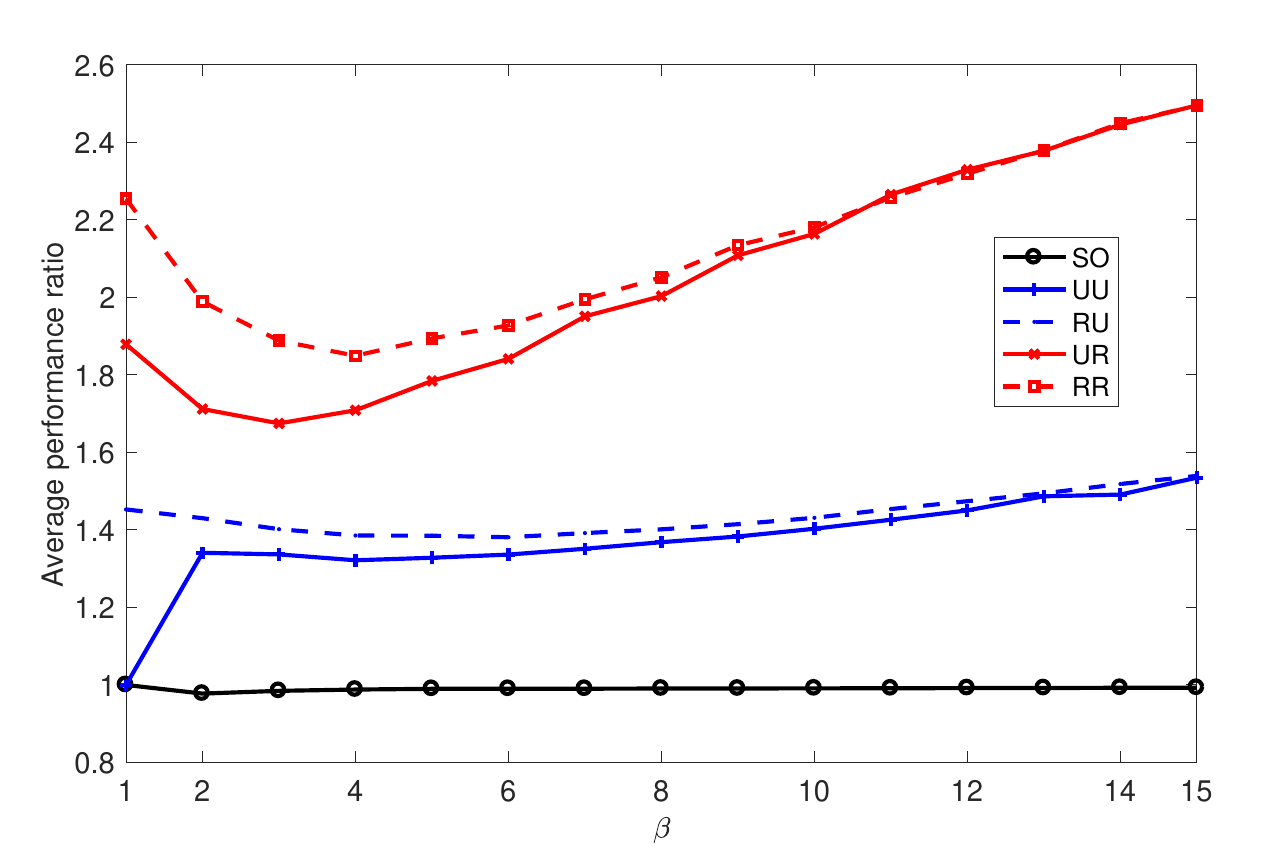}
	\caption{Average performance of \emph{AANC} versus SO,UU,RU,UR, and RR as a function of $\beta$ under the uniform distribution.} 
	%\vspace{-0.2in}
	\label{figure-uniform-nvsm-average-nonconcave}
\end{figure}

%\vspace{-0.1cm}
\subsubsection{Real-world utility functions}
\begin{figure}
	\centering 
	%\hspace{-0.2in}  
	\includegraphics[height=4cm,width=5.2cm]{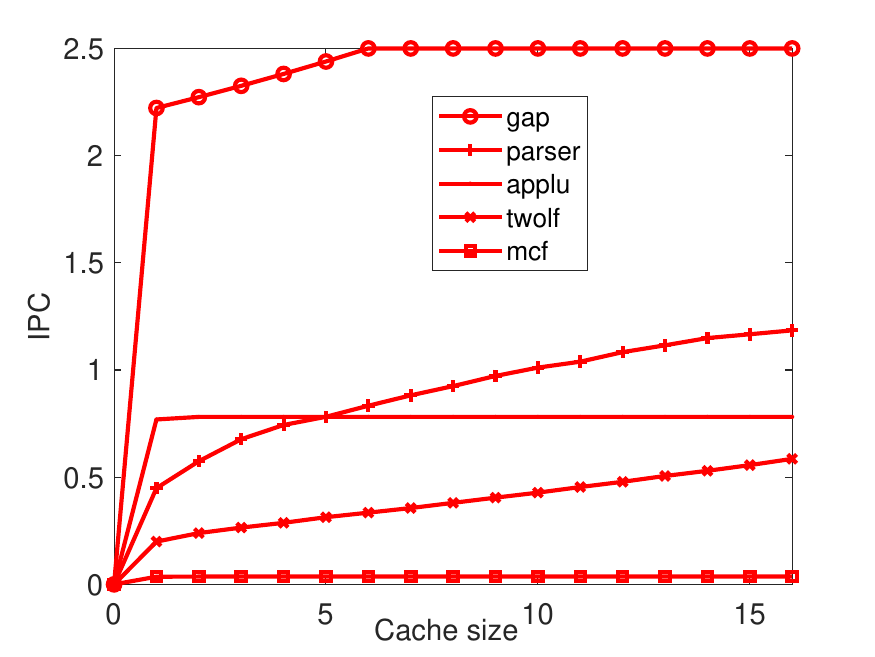}
	\caption{Concave utility functions for 5 SPEC benchmarks using data adapted from \cite{Qureshi, Lai}. The x-axis shows the normalized amount of cache an application is allocated, and the y-axis shows the IPC.} 
	%\vspace{-0.2in}
	\label{figure-real-concave-utility}
\end{figure}

\begin{figure}
	\centering 
	%\hspace{-0.2in}
	\subfigure{ 
		\label{figure-utility-set1} 
		\includegraphics[height=4cm,width=5.2cm]{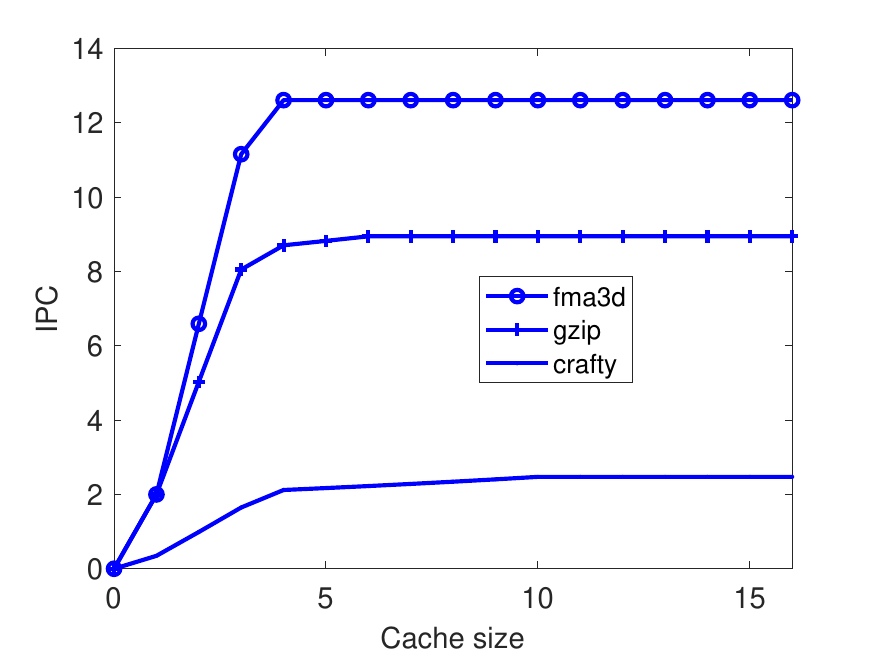}} 
	%\hspace{-0.05in} 
	\subfigure{ 
		\label{figure-utility-set2} 
		\includegraphics[height=4cm,width=5.2cm]{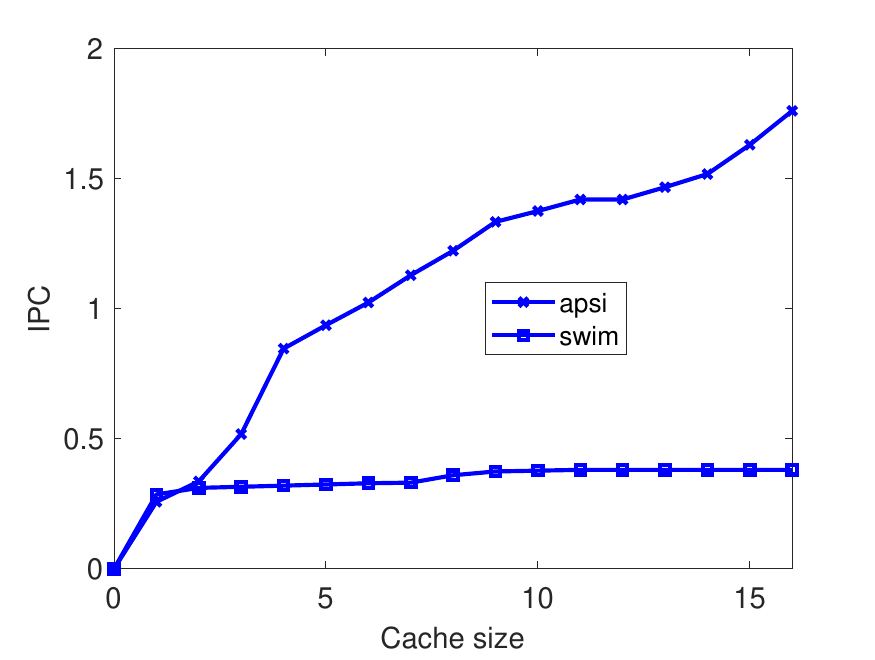}}    
	\caption{Nonconcave utility functions for 5 SPEC benchmarks using data adapted from \cite{Qureshi, Lai}.} 
	%\vspace{-0.2in}
	\label{figure-real-nonconcave-utility}
\end{figure}

We also use real-world utility functions derived from 10 SPEC CPU benchmarks: \emph{gap}, \emph{parser}, \emph{applu}, \emph{twolf}, \emph{mcf}, \emph{fma3d}, \emph{gzip}, \emph{crafty}, \emph{apsi}, \emph{swim}. Each benchmark's execution speed (utility) is measured using its IPC (instructions per cycle) as a function of the amount of L2 cache the program is allocated\footnote{In multicore processors with multiple levels of cache, the first level L1 cache is typically privately owned by each core, while the L2 cache is shared by different cores.}. Figure \ref{figure-real-concave-utility} shows the utility functions of the first 5 benchmarks with concave utility functions, and Figure \ref{figure-real-nonconcave-utility} shows the last 5 benchmarks with nonconcave utility functions \footnote{Our data is obtained from \cite{Qureshi, Lai}, which do not show the benchmarks' utility with zero resource. Since having sufficient cache is crucial to performance, we make the simplifying assumption that the utility of a benchmark allocated zero L2 cache is zero \cite{PanMakespan}. While this assumption does not hold in practice, we use it to facilitate a more uniform comparison with synthetic utility functions.}. In the experiment on real-world utility functions, we set the number of servers to be $m=2$, and the number of threads to be $n=8$, representing some 8 among the 10 benchmarks described earlier. We then test the effects of varying the server resource size $C$ representing the size of the L2 cache. 

\begin{table}\small\label{table-realistic-function-nonconcave2}
\caption{Performance of \emph{AANC} vs. SO, UU, RU, UR and RR for different resource sizes using the first set of 8 real-world utility functions (i.e., \emph{gap}, \emph{parser}, \emph{applu}, \emph{twolf}, \emph{mcf}, \emph{fma3d}, \emph{gzip}, \emph{crafty})}
%\caption{$\Delta_{th}$ depending on the relationship between $\bar{p}$ and $k$}
\centering
\begin{tabular}{|c|c|c|c|c|c|}
\hline
 Server resource size     &SO&  UU & RU & UR        & RR  \\
\hline
 4     &1 &  2.85  &1.87&12.13&2.04      \\
\hline
8      &0.98&1.61&1.55&2.14&1.62\\
\hline
16    &1&1.02&1.36&1.19&1.40\\
\hline
\end{tabular}
\end{table}%

We look at the performance using real-world utility functions from the first set of 8 CPU benchmarks (i.e., \emph{gap}, \emph{parser}, \emph{applu}, \emph{twolf}, \emph{mcf}, \emph{fma3d}, \emph{gzip}, \emph{crafty}). The first 5 benchmarks have concave utility functions and the last 3 ones have nonconcave utility functions. We set the number of threads to  $n=8$ to represent the 8 benchmarks. Table IV shows the ratio of \emph{AANC}'s total utility compared to the utilities of SO,UU,RU,UR, and RR for different resource sizes $C$. Compared to SO, our performance never drops below 0.98, i.e., \emph{AANC} always achieves at least 98\% of the optimal utility. The ratios of \emph{AANC}'s utility compared to those of UU, UR, RU and RR are again always above 1, so that it always perform better than the heuristics.  In addition, the performance ratio decreases as $C$ increases. For example, when $C=4$,  our algorithm is $2.85\times$, $1.87\times$, $12.13\times$ and $2.04\times$ better than UU,RU,UR,RR, respectively. For $C=16$, our algorithm is $1.02\times$, $1.36\times$, $1.19\times$ and $1.40\times$ better than UU,RU,UR,RR, respectively. This shows our algorithm is most effective in highly resource constrained environments or when running a large number of threads.

\begin{table}\small\label{table-realistic-function-nonconcave2}
\caption{Performance of \emph{AANC} vs. SO, UU, RU, UR and RR for different resource sizes using the second set of 8 real-world utility functions (i.e., \emph{gap}, \emph{applu}, \emph{twolf}, \emph{mcf}, \emph{fma3d}, \emph{gzip}, \emph{apsi}, \emph{swim})}
%\caption{$\Delta_{th}$ depending on the relationship between $\bar{p}$ and $k$}
\centering
\begin{tabular}{|c|c|c|c|c|c|}
\hline
 Server resource size     &SO&  UU & RU & UR        & RR  \\
\hline
 4     &1 &  2.94  &1.91&14.35&2.14      \\
\hline
8      &0.987&1.62&1.53&2.21&1.62\\
\hline
16    &0.98&1.02&1.32&1.17&1.36\\
\hline
\end{tabular}
\end{table}%

Lastly, we look at \emph{AANC}'s performance using real-world utility functions from the second set of 8 CPU benchmarks (\emph{gap}, \emph{applu}, \emph{twolf}, \emph{mcf}, \emph{fma3d}, \emph{gzip}, \emph{apsi}, \emph{swim}). The first 4 benchmarks have concave utility functions and the last 4 ones have nonconcave utility functions. We set the number of threads to be $n=8$, representing the following 8 benchmarks. Table V shows the ratio of \emph{AANC}'s total utility compared to the utilities of SO,UU,RU,UR, and RR for different resource sizes $C$. The behaviors are similar with that for the first set of 8 benchmarks.

%% file: multi-types-dp.tex
\blue{
	\section{Extension to Multiple Resource Types}\label{sec-multi-type-dp}
	In this section, we further generalize our algorithms to a setting where each server has multiple types of resources, and each thread's utility function depends on all the resource types.  An example of this situation is in cloud computing, where a virtual machine's (\emph{i.e.} thread's) performance is affected by both the number of CPU cores and amount of memory it is allocated.   We adapt our algorithms for the single resource setting to the multi-resource one, and demonstrate that the new algorithm achieves good empirical performance.  
	
	The model we use is similar to the one presented in Section III, and we only point out the differences.  We assume there are $d$ types of resources, and refer to the $i$'th type of resource as a \emph{type-$i$ resource}. Each server has $C_i$ amount of type-$i$ resource, where $C_i$ is a positive integer. Each thread $t_i$ has a $d$-dimensional utility function $f_i(x_1,\ldots, x_d)$, where $x_j\in [0,C_j]$ is the amount of type-$j$ resource $t_i$ is allocated. We assume that each thread's utility function is nondecreasing and concave. An \emph{assignment} is given by a vector $[(r_1, c_{1,1},\ldots, c_{1,d}), \ldots, (r_{n,1}, c_{1,1}, \ldots, c_{n,d})]$, indicating that each thread $t_i$ is allocated $c_{i,j}$ amount of type-$j$ resource for $j=1,\ldots, d$ on server $s_{r_i}$.
	
	We now present the algorithm for the AA problem with multiple resource types, which we call \emph{AAMR}. Similar to Definition \ref{def-super-opt}, we first define the \emph{super-optimal utility} and \emph{super-optimal allocation} of the problem under multiple resource types.  
	\begin{definition}\label{def-super-opt-multi-types}
		%Given an AA problem $A$ with $m$ servers each with $C$ amount of resources, and $n$ threads with utility functions $f_1, \ldots, f_n$, a \emph{super-optimal allocation} for $A$ is a set of values $\hat{c}_1, \ldots, \hat{c}_n$ that maximizes the quantity $\sum_{i=1}^n f_i(\hat{c}_i)$, subject to $\sum_{i=1}^n \hat{c}_i \leq mC$.  Call $\sum_{i=1}^n f_i(\hat{c}_i)$ the \emph{super-optimal utility} for $A$.
		Given an instance $A$ of the AA problem with $m$ servers each with $C_i$ amount of type-$i$ resource for $i=1,\ldots, d$, and $n$ threads with utility functions $f_1, \ldots, f_n$, consider the quantity
		\begin{equation*}
		\hat{F} = \max_{c_{i,j}, i\in [1,n], j\in [1,d]}\sum_{i=1}^{n}f(c_{i,1},\ldots,c_{i,d})
		\end{equation*}
		subject to $\sum_{i=1}^n c_{i,j} \leq mC_j$ for $j=1,\ldots, d$. Let $\hat{c}_{i,j}, i\in [1,n], j\in [1,d]$ be values for  $c_{i,j}, i\in [1,n], j\in [1,d]$, respectively, which achieve the optimum $\hat{F}$.  Then we call $\hat{F} = \sum_{i=1}^n f(\hat{c}_{i,1},\ldots,\hat{c}_{i,d})$ the \emph{super-optimal utility} of $A$, and $\hat{c}_{i,j}, i\in [1,n], j\in [1,d]$ the \emph{super-optimal allocation} for $A$. 
	\end{definition}

AAMR first finds the super-optimal allocation of the problem with multiple resource types, and then uses the allocation as an indicator to assign threads and allocate resource to the assigned threads.  Finding the super-optimal allocation is equivalent to finding an optimal allocation which maximizes the total utility of $n$ threads on a single server, given the threads' utility functions $f_1, \ldots, f_n$, and a resource capacity $\hat{C}_i=mC_i$ for each type-$i$ resource on the server.  We use a dynamic programming algorithm which we call \textsc{DPSOpt} to solve this problem.  Given $x_1\in [0,\hat{C}_1],\ldots, x_n\in [0,\hat{C}_n]$ and  $k \in [1,n]$, let $F_k(x_1,\ldots, x_d)$ be the maximum utility of the first $k$ threads when they are allocated $x_1$ amount of type-$1$ resource, $\ldots$, $x_d$ amount of type-$d$ resource.  Then we have
\begin{eqnarray*}\small
F_k(x_1,\ldots, x_d)\!\!\!&=&\max_{z_i\in [0,\min(x_i,C_i)], i=1,\ldots, d}\{f_k(z_1,\ldots, z_d)\\
                    & &+F_{k-1}(x_1-z_1,\ldots, x_d-z_d)\}.
\end{eqnarray*}
Note that $z_i \leq C_i$ for $i=1,\ldots, d$, since the domain of each thread's utility function is $[0,C_i]$ for a type-$i$ resource. We can find the maximum total utility $F_n(\hat{C}_1,\ldots, \hat{C}_d)$ by finding all $F_k(x_1,\ldots, x_d)$ in increasing lexicographical order of $k, x_1,\ldots, x_d$, and then use backtracking to find the optimal allocation for threads $t_n, \ldots, t_1$.  The time to solve the DP (dynamic programming) is $O(n C^{2d})$ for $n$ threads, where $C = \max_i \hat{C}_i$.  Note however that $d$ is typically small in practice, \emph{e.g.} $d = 3$ when considering processing, memory and network bandwidth as resources.  Thus, the DP can typically be solved in an acceptable amount of time for moderate values of $C$.  

We now give the pseudocode for AAMR. The input includes the original utility functions $f_1, \ldots, f_n$, and a super-optimal allocation $\hat{c}_{i,j}$ for $i=1,\ldots, n$ and $j=1, \ldots, d$ returned by \textsc{DPSOpt}. Note that we do not use linearized utility functions as we do for the single resource setting due to the difficulty of linearizing multi-dimensional utility functions. Variable $C_{i,j}$ represents the amount of type-$j$ resource remaining on server $i$, and $R$ is the set of unassigned threads. The outer loop of the algorithm runs until all threads in $R$ have been assigned. During each iteration, $U$ is the set of \emph{(thread, server)} pairs such that the server has at least as much remaining resource as the thread's super-optimal allocation, for every resource type. If any such pair exists, then in line 7 of AAMR we find a thread in $U$ with the greatest utility when given its super-optimal allocation. Otherwise, in line 10 we find a
thread which can obtain the greatest utility when running on any server and using the minimum value between the thread's super-optimal allocation and the remaining resources on the server, for each resource type.  In both cases we assign the thread in line 13 to a server giving it the greatest utility. Lastly, we update the server's remaining resources accordingly.

\begin{algorithm}\small	
	\caption{Pseudocode for AAMR}
	\begin{algorithmic}[1]
	%\Procedure{AlgMultiType}{}
		\State \textbf{Input}: Utility functions $f_1, \ldots, f_n$, and super-optimal allocation $\hat{c}_{i,j}$ for $i=1,\ldots, n$ and $j=1, \ldots, d$ returned by \textsc{DPSOpt}
		\State $C_{i,j} \gets C_j$ for $i=1, \ldots, m$ and $j=1, \ldots, d$ 
		\State $R \gets \{1, \ldots, n\}$
		\While{$R \neq \emptyset$}
		\State $U \gets \{ (i,j) \, | \, (i \in R) \wedge (1 \leq j \leq m) \wedge (\hat{c}_{i,1} \leq C_{j,1})\wedge\ldots \wedge (\hat{c}_{i,d} \leq C_{j,d}) \}$
		\If{$U \neq \emptyset$}
		\State \!\!\!\!\!\!\!\!\!\!\!$(i,j) \gets \argmax_{(i,j) \in U} \,  f_i(\hat{c}_{i,1},\ldots, \hat{c}_{i,d})$
		\State \!\!\!\!\!\!\!\!\!\!\!$c_{i,k} \gets \hat{c}_{i,k}$ for $k=1, \ldots, d$ 
		\Else
		\State \!\!\!\!\!\!\!\!\!\!\!$(i,j) \gets \argmax_{i \in R, 1 \leq j \leq m}\,$ $f_i(\min(C_{j,1},\hat{c}_{i,1}),\ldots,\min(C_{j,d},\hat{c}_{i,d}))$
		\State \!\!\!\!\!\!\!\!\!\!\!$c_{i,k} \gets \min(C_{j,k},\hat{c}_{i,k})$ for $k=1, \ldots, d$ 
		\EndIf
		\State $r_i \gets j$
		\State $R \gets R - \{i\}$
		\State $C_{j,k} \gets C_{j,k} - c_{i,k}$ for $k=1, \ldots, d$ 
		\EndWhile
		\State \textbf{return} $(r_1,c_{1,1},\ldots, c_{1,d}),\ldots,(r_n, c_{n,1},\ldots, c_{n,d})$
		%\EndProcedure
		%\Statex
		%\Procedure{DpAlgSin}{}
    %\State \textbf{Input}: Utility functions $f_1, \ldots, f_n$, and resource capacity for each resource type $\hat{C}_1, \ldots, \hat{C}_d$
		%\For{$x_1\in [0, \hat{C}_1], \ldots, x_d\in [0, \hat{C}_d]$}
		%\State $F_{1, i}(x_1,\ldots, x_d)\gets f_i(x_1,\ldots, x_d)$ for $i=1,\ldots, d$
		%\State $Q_{1, i}(x_1,\ldots, x_d)\gets x_i$ for $i=1,\ldots, d$
    %\EndFor		
		%\For{$k\gets 2$ to $n$}
		%\For{$x_1\in [0, \hat{C}_1], \ldots, x_d\in [0, \hat{C}_d]$}
		%\State $F_k(x_1,\ldots, x_d)\gets\max_{z_1\in [0,x_1],\ldots, z_d\in [0,x_d]}\{f_k(z_1,\ldots, z_d)+F_k(x_1-z_1,\ldots, x_d-z_d)\}$
		%\State $(z_1,\ldots, z_d)\gets\max_{z_1\in [0,x_1],\ldots, z_d\in [0,x_d]}\{f_k(z_1,\ldots, z_d)+F_k(x_1-z_1,\ldots, x_d-z_d)\}$
		%\State $Q_{k, j}(x_1,\ldots, x_d)\gets z_j$ for $j=1,\ldots, d$
		%\EndFor
		%\EndFor
		%\State $\hat{c}_{n,j}\gets Q_n(x_1,\ldots, x_d)$ for $j=1, \ldots, d$
		%\State $l_n \gets \hat{C}_j-\hat{c}_{n,j}$ for $j=1, \ldots, d$
		%\For{$k\gets n-1$ to $1$}
		%\State $\hat{c}_{k,j}\gets S_{k,j}(l_1,\ldots, l_d)$ for $j=1, \ldots, d$
		%\State $l_k \gets \hat{c}_{n,j}-\hat{c}_{k,j}$
		%\EndFor
    %\State \textbf{return} $\hat{c}_{i,j}$ for $i=1,\ldots, n$ and $j=1, \ldots, d$
%\EndProcedure
\end{algorithmic}
\label{Alg3}
\end{algorithm}

%We can analyze AAMR's time complexity.
%\begin{theorem}
%\label{thm-time-3}
%AAMR runs in $O(mn^2d+nm^{2d}(\prod_{i=1}^{d}C_i)^2$ time. 
%\end{theorem}

\subsection{Experimental Evaluation of AAMR}
As multi-dimensional utility functions pose more challenges than one-dimensional ones, the analysis of AAMR is performed numerically. Similar to Section VIII, we compare AAMR to the super-optimal (SO) utility, which is an upper bound on the optimal utility, and which can be computed using \textsc{DPSOpt}. We also compare AAMR with several simple but practical heuristics, including UU (uniform assignment and uniform allocation), UR (uniform assignment and random allocation), RU (random assignment and uniform allocation), RR (random assignment and random allocation).  We consider two resource types, and assume that each thread $t_i$ has a random two-dimensional concave utility function $f_i(x_1,x_2)=\gamma_1x^{\alpha_1}_1+\gamma_2x^{\alpha_2}_2$, where $x_1\in [0,C_1]$, $x_2\in [0,C_2]$ are variables, and $\gamma_1,\gamma_2>0$, $\alpha_1, \alpha_2\in (0,1)$ are parameters. We generate $\gamma_1, \gamma_2, \alpha_1, \alpha_2$ randomly according to the uniform distribution. In the experiment, we set the number of servers to be $m=4$ and the resource size to be $C_1=40, C_2=20$. We test the effects of varying parameter $\beta = \frac{n}{m}$, which represents the average number of threads per server. The following results show the average performance from 100 random trials. 

Figure \ref{figure-nvsm-average-multi-type} shows the average ratio of AAMR's total utility to the utilities of the other algorithms, for $\beta$ varying between 1 to 7. The behaviors shown in the figure are similar to those for one resource type in Section VIII. In addition, compared to SO, AAMR's utility ratio never drops below 0.96, indicating that AAMR always achieves at least 96\% of the optimal utility. Moreover, AAMR is always no worse than UU, RU, UR, RR, and it is $1.98\times$, $1.97\times$, $2.67\times$, $2.62\times$ better than UU, RU,UR, RR when $\beta=7$. 

\begin{figure}
	\centering 
	%\hspace{-0.2in}  
	\includegraphics[height=4cm,width=5.2cm]{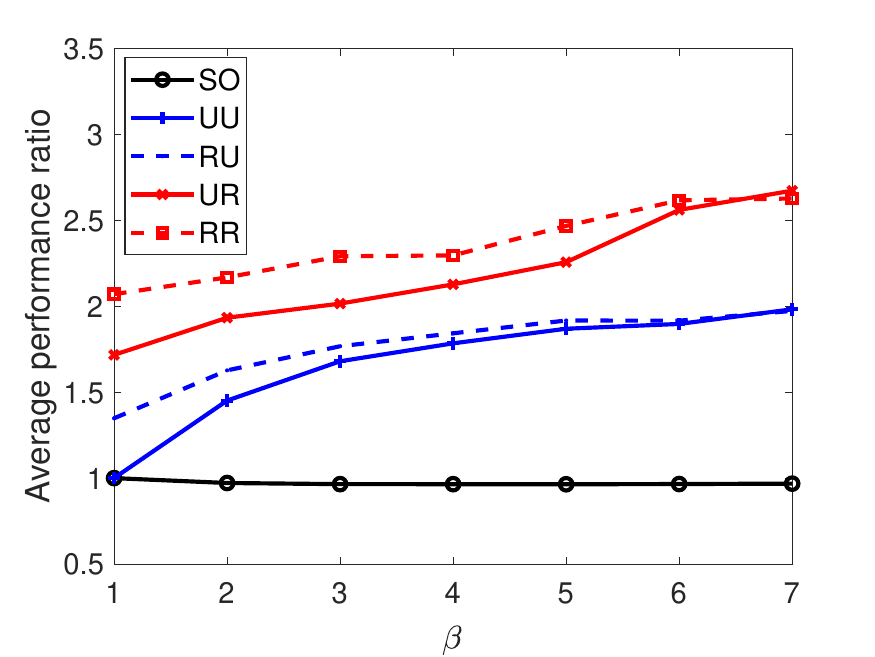}
	\caption{Average performance of Algorithm \emph{AAMR} versus SO,UU,RU,UR,RR as a function of $\beta$ under two resource types.} 
	%\vspace{-0.2in}
	\label{figure-nvsm-average-multi-type}
\end{figure}
}

%% file: conclusion.tex
\section{Conclusion}
In this paper, we studied the novel problem of simultaneously assigning threads to servers and allocating server resources to maximize total utility. We showed that the problem is NP-hard, even when there are only two servers and all utility functions are concave.  For concave utility functions, we presented two algorithms with approximation ratio $2 (\sqrt{2}-1) > 0.828$, running in times $O(mn^2 + n (\log mC)^2)$ and $O(n (\log mC)^2)$, respectively. In addition, we presented an algorithm with approximation ratio $\frac{1}{2}$ for threads with nonconcave utility functions, and an algorithm for concave utility functions with multiple resource types.  Lastly, we tested our algorithms on multiple types of threads, and found that our algorithms always achieve at least 92\% of the optimal utility, and typically over 98\% of the optimal utility. Our utility is up to 9 times better than those of several heuristic methods. 

In this work we considered homogeneous servers each with the same amount of resources.  We are interested in extending this model to accommodate heterogeneous servers with different capacities. \blue{In addition, our algorithms are currently centralized, and all decisions are made by a single scheduler process.  To scale the AA problem to larger system settings, we would like to consider distributed versions of our algorithms, where assignments and allocations are made concurrently by multiple schedulers.}  

%% file: author.tex
\section*{Author Biography}
\vspace{-0.5in}
\begin{IEEEbiography}[{\includegraphics[width=0.8in,height=0.8in,clip,keepaspectratio]{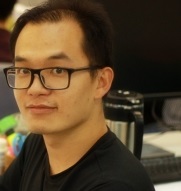}}]{Pan~Lai} received the PhD degree from School of Computer Engineering in Nanyang Technological University in 2016. He was a Postdoctoral Research Fellow in Singapore University of Technology and Design, Singapore during 2016-2019. His research interests include resource allocation and scheduling algorithm design in computer and network systems, network economics and game theory. 
\end{IEEEbiography}

%\textbf{Pan Lai} received his bachelor degree in Computer Engineering from University of Electronic Science and Technology of China in 2009. He obtained his PhD degree in Computer Science Division, School of Computer Engineering, Nanyang Technological University. He is currently a Postdoctoral Research Fellow with the Engineering Systems and Design Pillar, Singapore University of Technology and Design, Singapore. His research interest includes resource allocation and scheduling algorithm design in computer systems (e. g.,  multicore systems, cloud computing), parallel and distributed computing, network economics and game theory, machine learning and big data. 
\vspace{-0.6in}
\begin{IEEEbiography}[{\includegraphics[width=0.8in,height=0.8in,clip,keepaspectratio]{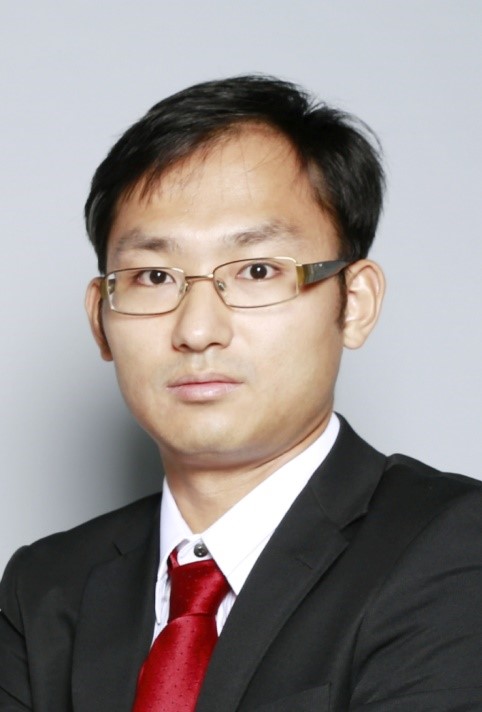}}] {Rui~Fan} is an associate professor in computer science in the School of Information Science and Technology at ShanghaiTech University.  He received his BSc from the California Institute of Technology in 2000, and his PhD in computer science from the Massachusetts Institute of Technology in 2008.  Prior to joining ShanghaiTech he was an assistant professor at Nanyang Technological University.  Dr. Fan's main research interests are in parallel and distributed computing, including efficient algorithms and optimizations for parallel architectures. His current work focuses on accelerating deep learning through improved network architectures, numerical optimization techniques and efficient implementations.
\end{IEEEbiography}

%\vspace{-0.3in}
%\begin{IEEEbiography}[{\includegraphics[width=1in,height=1.25in,clip,keepaspectratio]{fan.png}}] {Rui~Fan}received the B.Sc. degree in computer science and mathematics from Caltech, California, USA, in 2000 and the Ph.D. degree in computer science from Massachusetts Institute of Technology (MIT), Cambridge, USA in 2008. He was a Postdoc at the University of Toronto in 2008, and a Viterbi Postdoc at the Technion, Haifa, Israel in 2010. He was an Assistant Professor at Nanyang Technological University, Singapore. He is currently an Associate Professor at ShanghaiTech University, Shanghai, China. His research interests include a number of topics in parallel and distributed computing, the design of efficient algorithms and algorithmic solutions to problems on large datasets, and applications areas such as machine learning, energy efficient computing, and managing large scale systems. He is also involved in works on distributed algorithms and synchronization related issues as well as computational lower bounds. Dr. Fan’s Ph.D. dissertation received MIT’s Sprowls Award for the best computer science thesis. He also received two best student paper awards at the ACM Principles of Distributed Computing (PODC) conference.
%\end{IEEEbiography}

\vspace{-0.5in}
\begin{IEEEbiography}[{\includegraphics[width=0.8in,height=0.8in,clip,keepaspectratio]{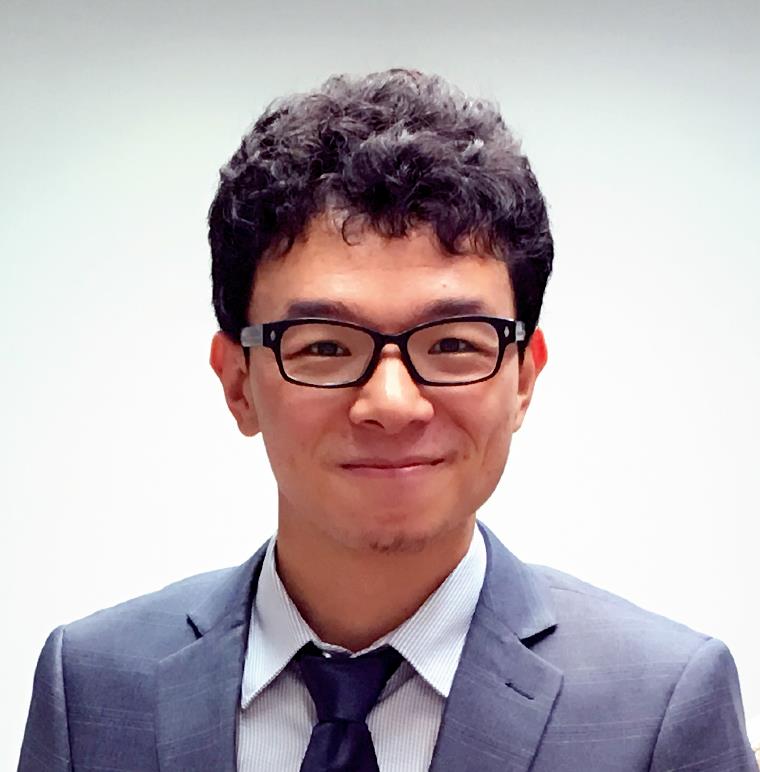}}]{Xiao Zhang} received PhD degree from City University of Hong Kong, Hong Kong, 2016. He is currently an associate professor in South-Central University for Nationalities, China. His research interests include algorithms design and analysis, combinatorial optimization, wireless, and UAV networking.
\end{IEEEbiography}

\vspace{-0.7in}
\begin{IEEEbiography}[{\includegraphics[width=0.6in,height=0.6in,clip,keepaspectratio]{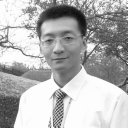}}]{Wei~Zhang} (S'10-M'16) received Ph.D. degree in computer science from Nanyang Technological University, Singapore, in 2015. He is currently an Assistant Professor with Singapore Institute of Technology. His current research interests include energy optimization as well as smart city.
\end{IEEEbiography}

\vspace{-0.7in}
\begin{IEEEbiography}[{\includegraphics[width=0.8in,height=0.8in,clip,keepaspectratio]{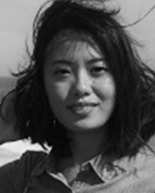}}]{Fang~Liu} (S'10-M'16) received the Ph.D. degree in computer science from Nanyang Technological University, Singapore, in 2015. Currently, she is a lecturer in Singapore University of Social Sciences, Singapore. Her research interests include secure data analytics, computational intelligence and so on.  
\end{IEEEbiography}